\documentclass{lmcs}
\pdfoutput=1
\usepackage[utf8]{inputenc}

\usepackage{lastpage}
\lmcsdoi{21}{1}{2}
\lmcsheading{}{\pageref{LastPage}}{}{}%
{Dec.~06,~2023}{Jan.~10,~2025}{}

\usepackage{graphicx}
\usepackage{microtype}
\usepackage{amssymb}
\usepackage{xcolor}
\usepackage{xspace}
\usepackage{amsmath}
\usepackage{stmaryrd}
\usepackage{hyperref}
\usepackage{proof}
\usepackage{url}
\usepackage{tikz}
\usepackage{wrapfig}
\usepackage{mathabx}
\usepackage{mathtools}
\usepackage{enumitem}
\usepackage{arydshln}

\usepackage[size=footnotesize,textwidth=2.6cm]{todonotes}

\newcommand{\instr}[2]{\stackrel{#1}{#2}}

\newcommand{\modif}[1]{#1}
\newcommand{\modiff}[1]{#1}
\newcommand{\modifRev}[1]{#1}

\newcommand{\added}[1]{#1}
\newcommand{\maude}{MAUDE}

\newcommand{\rollSafe}{roll-safe}

\newcommand{\cherry}{$\mathtt{cherry\text{-}pi}$}
\newcommand{\Cherry}{$\mathtt{cherry\text{-}pi}$}

\newcommand{\theoremApp}[1]{
\noindent
\textbf{Theorem  #1.}}

\newcommand{\pic}{$\pi$-calculus}

\newcommand{\typeT}{T}
\newcommand{\typeU}{U}
\newcommand{\typeV}{V}

\newcommand{\ifPi}{{\it\sf if}}
\newcommand{\thenPi}{{\it\sf then}}
\newcommand{\elsePi}{{\it\sf else}}

\newcommand{\ite}{\tau}
\newcommand{\commitLab}{\mathit{cmt}}
\newcommand{\rollLab}{\mathit{roll}}
\newcommand{\abortLab}{\mathit{abt}}
\newcommand{\commitType}{\mathtt{cmt}}
\newcommand{\rollType}{\mathtt{roll}}
\newcommand{\abortType}{\mathtt{abt}}

\newcommand{\ce}[1]{\bar{#1}} 

\newcommand{\requestAct}[3]{\ce{#1}(#2) . #3}
\newcommand{\acceptAct}[3]{#1(#2) . #3}

\newcommand{\requestPrefix}[2]{\overline{#1}(#2)}
\newcommand{\acceptPrefix}[2]{#1(#2)}

\newcommand{\send}[2]{#1!\langle#2\rangle}
\newcommand{\sendAct}[3]{\send{#1}{#2} .  #3}
\newcommand{\receive}[2]{#1?(#2)}
\newcommand{\receiveAct}[3]{\receive{#1}{#2} . #3}

\newcommand{\select}[2]{#1 \triangleleft #2}
\newcommand{\selectAct}[3]{\select{#1}{#2} . #3}
\newcommand{\branching}[1]{#1 \triangleright}
\newcommand{\branch}[2]{#1\, :\, #2}
\newcommand{\branchSep}{,}
\newcommand{\branchAct}[2]{\branching{#1} \{#2\}}

\newcommand{\ifthenelseAct}[3]{\ifPi\ #1\ \thenPi\ #2\ \elsePi\ #3}

\newcommand{\inact}{\mathbf{0}}

\newcommand{\recAct}[2]{\mu #1.#2}

\newcommand{\ctrue}{\texttt{true}}
\newcommand{\cfalse}{\texttt{false}}

\newcommand{\congr}{\equiv}
\newcommand{\subst}[2]{[ #1 / #2 ]}

\newcommand{\expreval}[2]{#1 \downarrow #2}

\newcommand{\rulelabel}[1]{[{\sc #1}]}

\newcommand{\actionLab}{\ell}
\newcommand{\typeLab}{\lambda}

\newcommand{\cmt}{\stackrel{cmt}{\fwred}}
\newcommand{\cmts}[1]{\stackrel{cmt \,#1}{\fwred}}
\newcommand{\rll}{\stackrel{roll}{\bwred}}
\newcommand{\rlls}[1]{\stackrel{roll \,#1}{\bwred}}

\newcommand{\abt}{\stackrel{abt}{\bwred}}

\newcommand{\fwred}{\twoheadrightarrow}
\newcommand{\bwred}{\rightsquigarrow}
\newcommand{\fwbwred}{\rightarrowtail}
\newcommand{\auxrel}[1]{\transition{#1}}
\newcommand{\typeTrans}[1]{\transition{#1}}
\newcommand{\typered}{\longmapsto}
\newcommand{\compliant}{\dashV}

\newcommand{\transition}[1]{\xrightarrow{#1}}

\newcommand{\boolType}{\mathtt{bool}}
\newcommand{\intType}{\mathtt{int}}
\newcommand{\strType}{\mathtt{str}}

\newcommand{\outType}[1]{![#1]}

\newcommand{\inpType}[1]{?[#1]}

\newcommand{\choiceType}{\oplus}
\newcommand{\selType}[1]{\triangleleft[#1]}

\newcommand{\branchType}[1]{\triangleright[#1]}

\newcommand{\selTypeLabel}[1]{\triangleleft\, #1}
\newcommand{\branchTypeLabel}[1]{\triangleright\, #1}
\newcommand{\inactType}{\mathtt{end}}
\newcommand{\errType}{\mathtt{err}}
\newcommand{\recType}[1]{\mu #1}

\newcommand{\sorting}{\Gamma}
\newcommand{\typing}{\Delta}
\newcommand{\basis}{\Theta}
\newcommand{\sessions}{A}
\newcommand{\comp}{\cdot}
\newcommand{\judge}{\ \vdash\ }
\newcommand{\hasType}{\ \blacktriangleright\ }
\newcommand{\reqOrApp}[1]{\hat{#1}}

\newcommand{\commit}{{\it\sf commit}}
\newcommand{\roll}{{\it\sf roll}}
\newcommand{\abort}{{\it\sf abort}}
\newcommand{\commitAct}[1]{\commit . #1}

\definecolor{mygray}{rgb}{.90,.90,.90}
\newcommand{\graybox}[1]{\colorbox{mygray}{#1}}
\newcommand{\gb}[1]{\graybox{$#1$}}

\newcommand{\singleSession}[2]{\big(\nu #1:#2\big)}
\newcommand{\logged}[2]{\langle #2 \rangle \blacktriangleright} 
\newcommand{\loggedShort}[2]{\langle #2 \rangle \!\blacktriangleright\!} 

\newcommand{\genSession}{k}
\newcommand{\sessId}{r}

\newcommand{\conf}[2]{\langle #1 \rangle \blacktriangleright #2} 
\newcommand{\initConf}[2]{( #1 ,  #2):} 
 
\newcommand{\confcomp}{\parallel}
\newcommand{\imposed}[1]{\underline{#1}} 
\newcommand{\checkpointType}[1]{\tilde{#1}} 

\newcommand{\comError}{\texttt{com\_error}}
\newcommand{\rollError}{\texttt{roll\_error}}

\newcommand{\collContext}{\mathbb{C}}

\newcommand{\BarbInp}[1]{\Downarrow_{#1?}}
\newcommand{\BarbOut}[1]{\Downarrow_{#1!}}
\newcommand{\BarbSelect}[2]{\Downarrow_{#1\triangleleft #2}}
\newcommand{\BarbBranching}[2]{\Downarrow_{#1\triangleright #2}}

\newcommand{\Barb}[1]{\Downarrow_{#1}}

\makeatletter
\newcount\dirtree@lvl
\newcount\dirtree@plvl
\newcount\dirtree@clvl
\def\dirtree@growth{%
  \ifnum\tikznumberofcurrentchild=1\relax
  \global\advance\dirtree@plvl by 1
  \expandafter\xdef\csname dirtree@p@\the\dirtree@plvl\endcsname{\the\dirtree@lvl}
  \fi
  \global\advance\dirtree@lvl by 1\relax
  \dirtree@clvl=\dirtree@lvl
  \advance\dirtree@clvl by -\csname dirtree@p@\the\dirtree@plvl\endcsname
  \pgf@xa=.32cm\relax
  \pgf@ya=-.32cm\relax
  \pgf@ya=\dirtree@clvl\pgf@ya
  \pgftransformshift{\pgfqpoint{\the\pgf@xa}{\the\pgf@ya}}%
  \ifnum\tikznumberofcurrentchild=\tikznumberofchildren
  \global\advance\dirtree@plvl by -1
  \fi
}

\tikzset{
  dirtree/.style={
    growth function=\dirtree@growth,
    every node/.style={anchor=north},
    every child node/.style={anchor=west},
    edge from parent path={(\tikzparentnode\tikzparentanchor) |- (\tikzchildnode\tikzchildanchor)}
  }
}
\makeatother

\usepackage{listings}
\begin{document}

\title[Checkpoint-Based Rollback Recovery in Session-based Programming]{Checkpoint-Based Rollback Recovery\\ in Session-based Programming}

\thanks{This research was funded in whole, or in part, 
by EPSRC EP/T006544/2, EP/K011715/1,  EP/K034413/1, EP/L00058X/1, EP/N027833/2, EP/N028201/1, EP/T014709/2, EP/Y005244/1, EP/V000462/1, EP/X015955/1, NCSS/EPSRC VeTSS and  Horizon EU TaRDIS 101093006,
Italian MUR PRIN 2020 project NiRvAna, the Italian MUR PRIN 2022 project DeKLA, by INdAM -- GNCS 2024 project \emph{Modelli composizionali per l'analisi di sistemi reversibili distribuiti (MARVEL)}, and the project SERICS (PE00000014) under the NRRP MUR program funded by the EU - NextGenerationEU}

\author[C. A. Mezzina]{Claudio Antares Mezzina\lmcsorcid{0000-0003-1556-2623}}[a]

\author[F. Tiezzi]{Francesco Tiezzi\lmcsorcid{0000-0003-4740-7521}}[b]

\author[N. Yoshida]{Nobuko Yoshida\lmcsorcid{0000-0002-3925-8557}}[c]

\address{Dipartimento di Scienze Pure e Applicate, Universit\`a di Urbino, Italy}
\address{Universit\`a degli Studi di Firenze, Italy}
\address{University of Oxford, United Kingdom}

\begin{abstract}
  To react to unforeseen circumstances or amend abnormal situations 
in communi\-cation-centric systems, 
programmers are in charge of  ``undoing'' the interactions 
which led to an undesired state. To assist this task, session-based
languages can be endowed with reversibility mechanisms.  
In this paper we propose a language enriched with programming facilities 
to \textit{commit} session interactions, to \textit{roll back} the computation to 
a previous commit point, and to \textit{abort} the session.
Rollbacks in our language always bring the system to previous 
visited states and  a rollback cannot bring the system back to a point prior 
to the last commit.
Programmers are relieved from the burden of ensuring that a rollback never 
restores a checkpoint imposed by a session participant different from 
the rollback requester. Such undesired situations are prevented at
design-time (statically) by relying on a decidable \emph{compliance} check 
at the type level, implemented in \maude. 
We show that the language satisfies error-freedom and progress of a session. 

\end{abstract}

\maketitle
\section{Introduction}
\label{sec:intro}

Reversible computing \cite{wg1,wg2} has gained interest for its application to different fields: 
from modelling biological/chemical phenomena~\cite{KuhnU18}, 
to simulation~\cite{PerumallaP13}, 
debugging~\cite{survey-of-reverse-debugging,GiachinoLM14,LaneseSuS22} 
and modelling fault-tolerant systems~\cite{DanosK05,LMSS13,VassorS18}. 
Our interest focuses on this latter application and stems from the fact
that reversibility can be used to rigorously model, implement and revisit programming 
abstractions for reliable software systems. 

Recent works~\cite{BarbaneraLd17,MezzinaP17a,MezzinaP21,CastellaniDG17,JLAMP_TY15} 
have studied the effect of reversibility in communication-centric scenarios, as a 
way to correct faulty computations by bringing back the system to a previous consistent state.
In this setting, processes' behaviours are strongly disciplined by their types, 
prescribing the actions they have to perform within a \textit{session}.
A session consists of a structured series of message exchanges, whose flow can be controlled 
via conditional choices, branching and recursion. Correctness of communication is statically 
guaranteed by a framework based on a (session) type discipline \cite{HuttelLVCCDMPRT16}. 
None of the aforementioned works 
addresses  systems in which the participants 
can \textit{explicitly} abort the session, commit a computation and roll it back to a previous checkpoint.
In this paper, we aim at filling this gap. 
We explain below the distinctive aspects of our checkpoint-based rollback recovery approach.
%

\paragraph{Linguistic primitives to explicitly  program reversible sessions.}
We introduce three primitives to: 
(i)~\emph{commit} a session, preventing undoing the interactions performed so far along the session; 
(ii)~\emph{roll back} a session, restoring the last saved process checkpoints;
(iii)~\emph{abort} a session, to discard the session, and hence all interactions already performed in it, 
thus allowing another session of the same protocol to start with possible different participants. 
Notice that most proposals in the literature (e.g., 
\cite{BarbaneraDd14,BarbaneraDLd16,BarbaneraLd17}) only consider an abstract view, as they focus on reversible 
contracts (i.e., types). Instead, we focus on programming primitives at process level, and use types 
for guaranteeing a safe and consistent system evolution.  

\paragraph{Asynchronous commits.}
Our commit primitive does not require a session-wide synchronisation among all participants, 
as it is a local decision. 
However, its effect is on the whole session, as it affects the other session participants. 
This means that each participant can independently decide when to commit. 
Such flexibility comes at the cost of being error-prone, especially considering that the 
programmer has not only to deal with the usual forward  executions, but also with the backward 
ones. Our type discipline allows for ruling out programs which may lead to these errors.
%
The key idea of our approach is that \emph{a session participant executing a rollback action is interested 
in restoring the last checkpoint he/she has committed}.
For the success of the rollback recovery it is irrelevant whether 
the `passive' participants go back to their own last checkpoints. Instead, if the `active' participant is  
unable to restore the last checkpoint he/she has created, because it has been replaced by a checkpoint 
imposed by another participant, the rollback recovery is considered unsatisfactory. 

\bigskip

In our framework, programmers are relieved from the burden of 
ensuring the satisfaction of rollbacks,  
since undesired situations are prevented at design time (statically) 
by relying on a \emph{compliance} check at the type level. 
%
%
To this end, we introduce \emph{\cherry}\
(\underline{che}ckpoint-based \underline{r}ollback \underline{r}ecover\underline{y} 
\underline{pi}-calculus), a variant of the session-based 
\pic~\cite{YoshidaV07} 
enriched with rollback recovery primitives. 
%
%
A key difference with respect to the standard binary type discipline is the \textit{relaxation} of the duality requirement. The types of two session 
participants are not required to be dual, but they will be compared with respect to a compliance 
relation (as in \cite{BarbaneraLd18}), which also 
takes into account the effects of commit and rollback actions. Such relaxation 
also involves the requirements concerning selection and branching types, 
and those concerning  branches of conditional choices.
The \cherry\ type system is used to infer types 
of session participants, which are then combined together for the compliance check. 

Reversibility in \Cherry\ is \textit{controlled} via two specific primitives:
a rollback one telling when a reverse computation has to take place, 
and a commit one limiting the scope of a potential reverse computation.
This implies that the calculus is not fully reversible (i.e., backward computations are not 
always enabled), leading to have properties that are relaxed and different with respect to other reversible 
calculi~\cite{DanosK04,CristescuKV13,LaneseMS10,JLAMP_TY15}.
We prove that \cherry\ satisfies the following properties:
(i) a rollback always brings back the system to a previous visited state
and (ii) it is not possible to bring the computation back to a point prior to the last checkpoint, which 
implies that our commits have a persistent effect. 
 %
%
%
Concerning soundness properties, we prove that (a)~our compliance check is decidable,
(b)~compliance-checked \cherry\ specifications never lead to communication errors (e.g., a blocked communication where there is a 
receiver without the corresponding sender), and (c)~compliance-checked \cherry\ specifications never activate undesirable rollbacks (according to 
our notion of rollback recovery mentioned above). Property (b) resembles the type safety property
of session-based calculi (see, e.g.,~\cite{YoshidaV07}), while property (c) is a new property specifically 
defined for \cherry. The technical development of property proofs turns out to be more intricate than 
that of standard properties of session-based calculi, due to the
combined use of type and compliance checking. 
%
To demonstrate feasibility and effectiveness of our rollback recovery approach, 
we have concretely implemented the compliance check using the \maude\ 
\cite{maude2007} framework \added{(the code is available 
at \url{https://github.com/tiezzi/cherry-pi}).}


\paragraph{Outline.} 
Section~\ref{sec:motivation} illustrates the key idea of our rollback 
recovery approach.
Section~\ref{sec:cherry} introduces the \cherry\ calculus.
Section~\ref{sec:rollback_safety} introduces typing and compliance 
checking.
\modif{Section~\ref{sec:maude} illustrates the \maude{} implementation of the compliance checking.}
Section~\ref{sec:properties} presents the properties satisfied by \cherry.
\modif{Section~\ref{sec:case_study} shows the application of the \cherry\ approach to a speculative execution scenario.}
\modif{Section~\ref{sec:rw} discusses  related work. Finally,} 
Section~\ref{sec:conclusion} concludes 
the paper with future work. 
%
\modif{The appendix reports on the omitted proofs. 
}
%

\medskip

This paper is a revised and enhanced version of \cite{coordination}. In particular:
\begin{itemize}
\item Sec.~\ref{sec:cherry} has been extended with rules in Fig. \ref{fig:semantics_cherry_pi_aux} and \ref{fig:struct}, which are omitted in \cite{coordination}, to provide a complete account of the \cherry\ semantics.
\item Sec.~\ref{sec:rollback_safety} has been extended with omitted rules (Fig. \ref{fig:typeSyntax_cherry} and \ref{fig:typingSystem_exp} are new, while Fig. \ref{fig:typingSystem_proc} and \ref{fig:typeSemantics_ext} have been extended), to provide a full account of the \cherry\ typing discipline. Moreover, the section includes the proof of Theorem~\ref{th:decidability}.
\item Sec.~\ref{sec:properties} has been extended by including full proofs of the results regarding the properties of \cherry.
\item Sec.~\ref{sec:case_study} is new. It shows the \cherry\ approach at work on a new scenario to provide a better understanding of the practical application of \cherry\ and, in particular, of the \maude{} implementation of its type semantics.
\item Sec.~\ref{sec:rw} has been revised and expanded, including the discussion of more recent related work and a table (Tab.~\ref{tab:overview}) providing a comparison of the related approaches in the literature. 
\item Finally, more commentary and explanations have been added throughout the paper, and the whole presentation has been carefully refined. 
\end{itemize}

\section{A reversible video on demand service example}
\label{sec:motivation}
\label{sec:motivating_example}
We discuss the motivations underlying our work by introducing our running example, 
a Video on Demand (VOD) scenario.   
The key idea is that a rollback requester is satisfied only if her restored checkpoint 
was set by herself. 
\begin{figure}[t]
\hspace*{-.5cm}
\includegraphics[scale=.32]{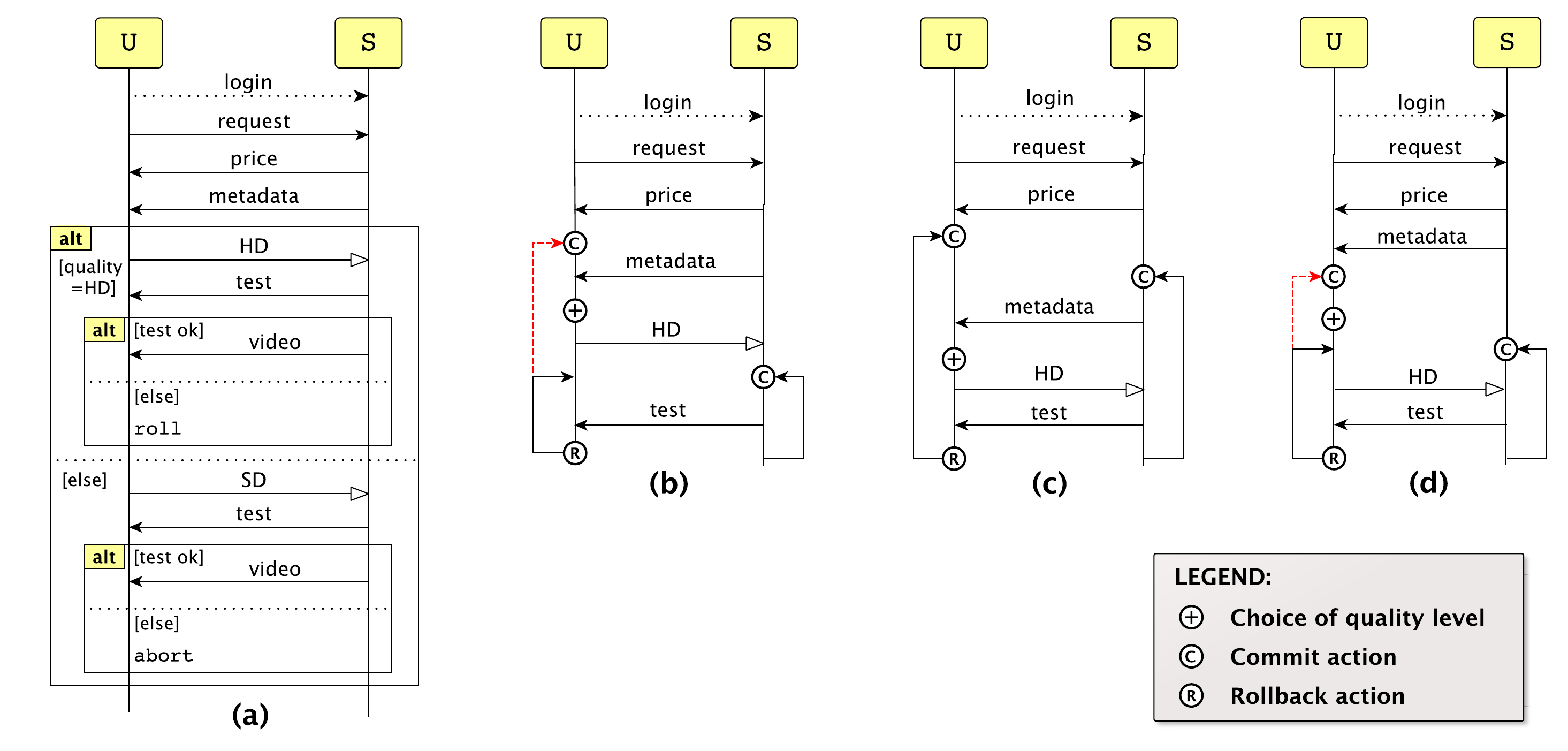}
\caption{VOD example: (a) a full description without commit actions; (b,d) runs with undesired rollback; 
(c) a run with satisfactory rollback. }
\label{fig:scenario}
\end{figure}
%
In Fig.~\ref{fig:scenario}(a), a service (\texttt{S}) offers to a user (\texttt{U}) videos with two different quality levels, namely high definition ({\small\textsf{HD}}) and standard definition ({\small\textsf{SD}}). 
After the \emph{login}, \texttt{U} sends her video \emph{request}, and receives the corresponding \emph{price} and \emph{metadata} (actors, directors, description, etc.) from \texttt{S}. According to this information, \texttt{U} selects the video quality. Then, she receives, first, a short \emph{test} video (to check the audio and video quality in her device) and, finally, the requested \emph{video}. If the vision of the {\small\textsf{HD}} test video is not satisfactory, \texttt{U} can roll back to her last checkpoint to possibly change the video quality, instead in the {\small\textsf{SD}} case \texttt{U} can abort the session. 

Let us now add commit actions as in the run shown in Fig.~\ref{fig:scenario}(b). 
After receiving the price, \texttt{U} commits, while \texttt{S} commits after the quality selection.  
In this scenario, however, if \texttt{U} activates the rollback, she is unable 
to go back to the checkpoint she set with her commit action because the actual 
effect of rollback is to restore the checkpoint set by the commit action performed by \texttt{S}.
\modif{Hence, \texttt{U} cannot use the rollback mechanism to undo her video quality choice and select the {\small\textsf{SD}} video.}

In the scenario in Fig.~\ref{fig:scenario}(c), instead, \texttt{S} commits after sending 
the price to \texttt{U}. In this case, no matter who first performed the commit action, the 
rollback results to be satisfactory. Also if \texttt{S} commits later, the checkpoint of 
 \texttt{U} remains unchanged, as \texttt{U} performed no other action between the two commits.
This would not be the case if both \texttt{U} and \texttt{S} committed after the communication of 
the metadata, as in Fig.~\ref{fig:scenario}(d). If \texttt{S} commits before \texttt{U}, 
no rollback issue arises, but if \texttt{U} commits first it may happen that her internal decision 
is taken before \texttt{S} commits. In this case, \texttt{U} would not be able to go back to the 
checkpoint set by herself, and she would be unable to change the video quality.   
 
These undesired rollbacks are caused by bad choices of commit points. 
We propose a compliance check that identifies  
these situations at design time.
%
\modifRev{Notably, our aim is to provide programming constructs and underlying mechanisms 
that would be easy to use and understand for the user. Thus, although it would be possible 
to define more expressive constructs, such as commit and rollback actions specifying 
a label as a parameter (in a way similar to the approaches introduced in \cite{LaneseMSS11,GiachinoLMT17}), we preferred to avoid them as they would make the system's behavior more intricate. 
In fact, such labelled actions would behave as a sort of \texttt{goto} jump, which 
can lead to unmanageable spaghetti code.}

\section{The \cherry\ calculus}\label{sec:cherry}

In this section, we introduce \cherry, 
a calculus (extending that in~\cite{YoshidaV07}) 
devised for studying 
sessions equipped with our checkpoint-based rollback recovery mechanism.

\subsection{Syntax}
The syntax of the \cherry\ calculus relies on the following base sets: 
\emph{shared channels} (ranged over by $a$), used to initiate sessions;
\emph{session channels} (ranged over by $s$), consisting of pairs of \emph{endpoints} 
(ranged over, with a slight abuse of notation, by $s$, $\ce{s}$) used by the two parties to interact within an established session;
\emph{labels} (ranged over by $l$), used to select and offer branching choices;
\emph{values} (ranged over by $v$), including booleans, integers and strings
(whose \emph{sorts}, ranged over by $S$, are $\boolType$, $\intType$ and $\strType$, respectively), 
which are exchanged within a session; 
\emph{variables} (ranged over by $x$, $y$, $z$), storing values and session endpoints;
 \emph{process variables} (ranged over by $X$), used for recursion. 

\emph{Collaborations}, ranged over by $C$, 
are given by the grammar in Fig.~\ref{fig:syntax_cherry_pi}. 
The key ingredient of the calculus is the set of actions for controlling the session rollback. 
Actions $\commit$, $\roll$ and $\abort$ are used, respectively, to commit a session 
(producing a checkpoint for each session participant), 
to trigger the session rollback (restoring the last committed checkpoints) 
or to abort the whole session. We discuss below the other constructs of the calculus, 
which are those typically used for session-based programming \cite{HondaVK98}.  
A \cherry\ collaboration is a collection \modifRev{(more specifically, a parallel composition)} of \emph{session initiators}, i.e. terms ready to initiate sessions 
by synchronising on shared channels.
A synchronisation of two initiators $\requestAct{a}{x}{P}$ and $\acceptAct{a}{y}{Q}$
causes the generation of a fresh session channel, whose endpoints replace 
variables $x$ and $y$ in order to be used by the triggered processes $P$ and $Q$, 
respectively, for later communications. No subordinate sessions can be initiated within a running session.

\begin{figure}[t]
\centering
	\begin{tabular}{@{}r@{\ }c@{\ }l@{\qquad }l@{}}
	$C$ & ::= & & \textbf{Collaborations} \\
	&             & $\requestAct{a}{x}{P}$ \ $\mid$ \ $ \acceptAct{a}{x}{P}$ 
	                   \ $\mid$ \
	                   $C_1 \!\mid\! C_2$ 
	                   & \ \ request, accept, parallel
	\\[.3cm]
	$P$ & ::= & & \textbf{Processes} \\
	&             & $\sendAct{x}{e}{P}$ 
		        \ $\mid$ \ $\receiveAct{x}{y:S}{P}$ 
		        & \ \ output, input
	\\
	& $\mid$ & $\selectAct{x}{l}{P}$ 
		        \ $\mid$ \ $ \branchAct{x}{\branch{l_1}{P_1} \branchSep \ldots \branchSep \branch{l_n}{P_n}}$  
		        & \ \ selection, branching
	\\	
	& $\mid$ & $\ifthenelseAct{e}{P_1}{P_2}$ 
	                   \ $\mid$ \ $X$ 
	                   \ $\mid$ \ $\recAct{X}{P}$
	                   \ $\mid$ \ $\inact$
	                   & \ \ choice, recursion, inact 
	\\
	& $\mid$ &  $\commitAct{P}$
	\ $\mid$ \ $\roll$
	\ $\mid$ \ $\abort$  
        & \ \ commit, roll, abort                   
	\\[.3cm]	
	$e$ & ::= &\ \  $v$ \ $\mid$ \ $+(e_1,e_2)$ \ $\mid$ \ $\wedge(e_1,e_2)$ \ $\mid$ \ \ldots  
	& \textbf{Expressions} 
	\\[.1cm]
	\hline
	\end{tabular}
	\caption{\Cherry\ syntax.}
	\label{fig:syntax_cherry_pi}
\end{figure}

When a session is started, each participant executes a \emph{process}. Processes are built up from the  empty process $\inact$ 
\modifRev{(which can do nothing)}
and basic actions by means of 
action prefix $\_\,.\,\_\,$ \modifRev{(which allows the process on the right of the $.$ operator to proceed 
once the action on the left of the $.$ operator is executed)},
conditional choice $\ifthenelseAct{e}{\_}{\_}$ \modifRev{(which has the usual meaning)}, 
and recursion $\recAct{X}{\_}$ \modifRev{(which behaves as its process argument where the occurrences 
of the process variable $X$ are replaced by the recursion process itself)}.
Actions $\send{x}{e}$ and \mbox{$\receive{y}{z:S}$}  denote output and input 
via session endpoints replacing $x$ and $y$, respectively. 
These communication primitives realise the standard synchronous 
message passing, where messages result from the evaluation of \emph{expressions},
which are defined by means of standard operators on boolean, integer and string values.
Variables that are arguments of input actions are (statically) typed by sorts. 
There is no need for statically typing the variables occurring as arguments of session initiating actions, as they are always 
replaced by session endpoints. 
Notice that in \cherry\ the exchanged values cannot be endpoints, 
meaning that session delegation (i.e., channel-passing) is not 
considered\footnote{\modifRev{Notably, even if session delegation is not supported, 
we cannot just consider single binary sessions avoiding the notion of 
collaborations. In fact, collaborations allows us to consider 
non-determinism at the level of session establishment (e.g., think of a client 
and two servers providing the same service).}}.
Actions $\select{x}{l}$ and $\branchAct{x}{\branch{l_1}{P_1} \branchSep \ldots \branchSep \branch{l_n}{P_n}}$ denote selection and branching respectively (where \linebreak[5] $l_1$, \ldots, $l_n$ are pairwise distinct). 

\begin{exa}\label{example_syntax}
Let us consider the VOD example informally introduced in Sec.~\ref{sec:motivating_example}.
The scenario described in Fig.~\ref{fig:scenario}(a) with commit actions placed as in 
Fig.~\ref{fig:scenario}(b) is rendered in \cherry\ as
$
C_{\texttt{US}} \ = \ \requestPrefix{login}{x}.\ P_\texttt{U} \ \mid\ \acceptPrefix{login}{y}.\ P_\texttt{S}
$,
where:
$$
\begin{array}{r@{\ \ }c@{ \ \ }l}
P_\texttt{U} & = & 
\sendAct{x}{v_{\textit{\textsf{req}}}}{\,}
\receiveAct{x}{x_{\textsf{\textit{price}}}:\intType}{\,}
\commitAct{\,}
\receiveAct{x}{x_{\textit{\textsf{meta}}}:\strType}{\,}
\ifPi\
(f_{\textit{\textsf{eval}}}(x_{\textit{\textsf{price}}},x_{\textit{\textsf{meta}}}))
\\
&&
\quad
\thenPi\ \selectAct{{x}}{l_{\textit{\textsf{HD}}}}{\,}
\receiveAct{x}{x_{\textit{\textsf{testHD}}}:\strType}{\,}\\
&&
\qquad\quad
(\ifPi\ (f_{\textit{\textsf{HD}}}(x_{\textit{\textsf{testHD}}}))\ \thenPi\ \receiveAct{x}{x_{\textit{\textsf{videoHD}}}:\strType}{\inact}\ \elsePi\ \roll)
\\
&&
\quad
\elsePi\ \selectAct{{x}}{l_{\textit{\textsf{SD}}}}{\,}
\receiveAct{x}{x_{\textit{\textsf{testSD}}}:\strType}{\,}\\
&&
\qquad\quad
(\ifPi\ (f_{\textit{\textsf{SD}}}(x_{\textit{\textsf{testSD}}}))\ \thenPi\ \receiveAct{x}{x_{\textit{\textsf{videoSD}}}:\strType}{\inact}\ \elsePi\ \abort)
\\[.5cm]
P_\texttt{S} & = & 
\receiveAct{y}{y_{\textit{\textsf{req}}}:\strType}{\,}
\sendAct{y}{f_{\textit{\textsf{price}}}(y_{\textit{\textsf{req}}})}{\,}
\sendAct{y}{f_{\textit{\textsf{meta}}}(y_{\textit{\textsf{req}}})}{\,}
\\
&&
\branchAct{y}{\
\branch{l_{\textit{\textsf{HD}}}\!}{\!\commitAct{\,}
\sendAct{y}{f_{\textit{\textsf{testHD}}}(y_{\textit{\textsf{req}}})}{\,}
\sendAct{y}{f_{\textit{\textsf{videoHD}}}(y_{\textit{\textsf{req}}})}{\,}
\inact}
\ \branchSep \\
&&
\qquad\ \
\branch{l_{\textit{\textsf{SD}}}\!}{\!\commitAct{\,}
\sendAct{y}{f_{\textit{\textsf{testSD}}}(y_{\textit{\textsf{req}}})}{\,}
\sendAct{y}{f_{\textit{\textsf{videoSD}}}(y_{\textit{\textsf{req}}})}{\,}
\inact}\
}
\end{array}
$$
Notice that expressions used for decisions and computations are abstracted by 
relations $f_{\textit{\textsf{n}}}(\cdot)$, whose definitions are left unspecified.
Considering the placement of commit actions depicted in Fig.~\ref{fig:scenario}(c), 
the \cherry\ specification of the service's process becomes:
$$
\begin{array}{l}
\receiveAct{y}{y_{\textit{\textsf{req}}}:\strType}{\,}
\sendAct{y}{f_{\textit{\textsf{price}}}(y_{\textit{\textsf{req}}})}{\,}
\commitAct{\,}
\sendAct{y}{f_{\textit{\textsf{meta}}}(y_{\textit{\textsf{req}}})}{\,}
\\
\branchAct{y}{\
\branch{l_{\textit{\textsf{HD}}}\!}{\!
\sendAct{y}{f_{\textit{\textsf{testHD}}}(y_{\textit{\textsf{req}}})}{\,}
\sendAct{y}{f_{\textit{\textsf{videoHD}}}(y_{\textit{\textsf{req}}})}{\,}
\inact}
\ \branchSep \\
\qquad\ \
\branch{l_{\textit{\textsf{SD}}}\!}{\!
\sendAct{y}{f_{\textit{\textsf{testSD}}}(y_{\textit{\textsf{req}}})}{\,}
\sendAct{y}{f_{\textit{\textsf{videoSD}}}(y_{\textit{\textsf{req}}})}{\,}
\inact}\
}
\end{array}
$$

Finally, considering the placement of commit actions depicted in Fig.~\ref{fig:scenario}(d), 
the \cherry\ specification of the user's process becomes:
$$
\begin{array}{l}
\sendAct{x}{v_{\textit{\textsf{req}}}}{\,}
\receiveAct{x}{x_{\textsf{\textit{price}}}:\intType}{\,}
\receiveAct{x}{x_{\textit{\textsf{meta}}}:\strType}{\,}
\commitAct{\,}
\ifPi\ (f_{\textit{\textsf{eval}}}(x_{\textit{\textsf{price}}},x_{\textit{\textsf{meta}}}))\
\thenPi\ \ldots
\end{array}
$$
\end{exa}
\medskip

\subsection{Semantics}
The operational semantics of \cherry\ is defined for \emph{runtime} terms, generated by the extended syntax of the calculus in Fig.~\ref{fig:syntax_cherry_pi_ext} (new constructs are highlighted by a grey 
background).
%
We use \modifRev{$\genSession$ to denote \emph{generic session endpoints}, i.e. $s$ or $\ce{s}$, and}
$\sessId$ to denote \emph{session identifiers}, i.e. session endpoints and variables.
Those runtime terms that can be also generated by the grammar in Fig.~\ref{fig:syntax_cherry_pi}
are called \emph{initial collaborations}.

At collaboration level, two constructs are introduced:
\mbox{$\singleSession{s}{C_1}\ C_2$ } represents a \emph{session} along the channel $s$ 
with associated starting checkpoint $C_1$ (corresponding to the collaboration that has initialised the session)
and code $C_2$;  \mbox{$\logged{\genSession}{P_1}P_2$} represents a \emph{log}
storing the checkpoint $P_1$ associated to the code $P_2$.
At process level, the only difference is that session identifiers $\sessId$ are used as first argument of communicating actions.  
\begin{figure}[t]
\centering
	\begin{tabular}{@{}r@{\ \ }c@{\ \ }l@{\ \ \ }l@{}}
	$C$ & ::= & $\requestAct{a}{x}{P}$ \ $\mid$ \ $ \acceptAct{a}{x}{P}$ 
	                   \ $\mid$ \
	                   $C_1 \!\mid\! C_2$ 
	                   \ $\mid$ \
	                   \graybox{$\singleSession{s}{C_1}\ C_2$}
	                   \ $\mid$ \ \graybox{$\logged{\genSession}{P_1}P_2$}    	                   
	                   & \textbf{Collaborations}
	\\[.2cm]
	$P$ & ::= & $\sendAct{\gb{\sessId\!}}{e}{P}$
		         $\mid$ \ $\receiveAct{\gb{\sessId\!}}{y:S}{P}$	
		         $\mid$ \ $\selectAct{\gb{\sessId\!}\!}{l}{P}$ 
		         $\mid$ \ $\branchAct{\gb{\sessId\!}\!}{\branch{l_1\!\!}{\!\!P_1} 
		                          \branchSep\! \ldots\! \branchSep
                                          \branch{l_n\!\!}{\!\!P_n}}
                                        \mid \cdots$
			& \textbf{Processes} 
	\\[.1cm]
	\hline
	\end{tabular}
	\caption{\Cherry\ runtime syntax (the rest of processes $P$ and 
expressions $e$ are as in Fig.~\ref{fig:syntax_cherry_pi}).}
	\label{fig:syntax_cherry_pi_ext}
\end{figure}

\modifRev{\textit{Bindings} are defined as follows: 
$\requestAct{a}{x}{P}$, $\acceptAct{a}{x}{P}$, and
\mbox{$\receiveAct{\sessId}{x:S}{P}$}
bind variable $x$ in $P$;
\mbox{$\singleSession{s}{C_1}\ C_2$} binds session endpoints $s$ and $\ce{s}$ in $C_2$
(in this respect, it acts similarly to the restriction of \pic, but its scope cannot be extended/extruded 
to avoid involving processes that are not part of the session in the rollback effect);
and $\recAct{X}{P}$ binds process variable $X$ in $P$. 
The occurrence of a name (where name stand for variable, process variable and session endpoint)
is \textit{free} if it is not bound; we assume that bound names are pairwise distinct. 
Two terms are \textit{alpha-equivalent} if one can be obtained from the other by consistently 
renaming bound names; as usual, we identify terms up to alpha-equivalence. 
Communication gives rise to \textit{substitutions} of variables with values:
we denote with $P\subst{v}{x}$ the process obtained by replacing each free 
occurrence of the variable $x$ in $P$ by the value $v$. 
Similarly, $P\subst{Q}{X}$ (resp. $P\subst{\genSession}{x}$)
denotes the process obtained by replacing each free 
occurrence of $X$ (resp. $x$) in $P$ by the process $Q$ (resp. generic session identifier $\genSession$).
The semantics of the calculus is defined for \emph{closed} terms, i.e. terms without free 
variables and process variables. 
}



Not all processes allowed by the extended syntax correspond to meaningful
collaborations. In a general term the processes stored in logs 
may not be consistent with the computation that has taken place.
We get rid of such malformed terms, as we
will only consider those runtime terms, called  \emph{reachable} collaborations, 
obtained by means of reductions from initial collaborations.
The operational semantics of \cherry\ is given in terms of a standard 
\emph{structural congruence} $\congr$ (\modif{given in Fig.~\ref{fig:struct}})
and a \emph{reduction} relation $\fwbwred$ 
given as the union of the \emph{forward reduction} relation $\fwred$ 
and \emph{backward reduction} relation $\bwred$.
%
%
The definition of the relation $\fwred$ over closed 
collaborations relies on an auxiliary labelled relation $\auxrel{\actionLab}$ over processes that specifies 
the actions that processes can initially perform and the continuation process obtained after each such action. 
\modif{We consider all reduction relations closed under structural congruence.}
%
%
Given a reduction relation $\mathcal{R}$, we will indicate with  $\mathcal{R}^{+}$ and  
$\mathcal{R}^{*}$ respectively the \textit{transitive} and the \textit{reflexive-transitive} closure of $\mathcal{R}$.

The operational rules defining the auxiliary labelled relation are in Fig.~\ref{fig:semantics_cherry_pi_aux}. 
Action label $\actionLab$ stands for either 
$\send{k}{v}$,
$\receive{k}{x}$,
$\select{k}{l}$,
$\branching{k}l$, 
$\commitLab$,
$\rollLab$,
$\abortLab$,
or $\ite$.
The meaning of the rules is straightforward, as they just produce as labels the actions currently 
enabled in the process. In doing that, expressions of sending actions and conditional choices are 
evaluated (auxiliary function \mbox{$\expreval{e}{v}$} says that closed expression $e$ evaluates to 
value $v$). 

\begin{figure*}[t]
	\centering
	\begin{tabular}{@{}l@{\quad\ \ }r@{}}
	$\sendAct{k}{e}{P} \auxrel{\send{k}{v}} P$
	\ \ ($\expreval{e}{v}$) 
	\  \rulelabel{P-Snd}
	&
	$\receiveAct{k}{x:S}{P} \auxrel{\receive{k}{x}} P$
	\  \rulelabel{P-Rcv}
	\\[.3cm]
	$\selectAct{{k}}{l}{P} \auxrel{\select{k}{l}} P$
	\  \rulelabel{P-Sel}
	&
	$\branchAct{k}{\branch{l_1\!\!}{\!\!P_1}  \branchSep \ldots \branchSep \branch{l_n\!\!}{\!\!P_n}}
	\auxrel{\branching{k}l_i} P_i$
	\ ($1 \!\leq\! i \!\leq\! n$) 
	\ \rulelabel{P-Brn}	
	\\[.3cm]
\multicolumn{2}{c}{
\modif{
	\raisebox{.4cm}{
	\begin{tabular}{@{}l@{}}
	$\ifthenelseAct{e}{P_1}{P_2} \auxrel{\ite} P_1$
	\quad ($\expreval{e}{\ctrue}$) 
	\\[.3cm]
	$\ifthenelseAct{e}{P_1}{P_2} \auxrel{\ite} P_2$ 
	\quad ($\expreval{e}{\cfalse}$) 
	\end{tabular}
	}
	\quad
	\raisebox{.4cm}{
	\begin{tabular}{@{\!}l}
	\rulelabel{P-IfT} 
	\\[.2cm]
	\rulelabel{P-IfF}
	\end{tabular}
         }
         \hspace{-.44cm}
}
}
	\\[.3cm]	
	$\commitAct{P} \auxrel{\commitLab} P$ 
	\ \ \rulelabel{P-Cmt}
	&
	$\roll \auxrel{\rollLab} \inact$ 
	\ \ \rulelabel{P-Rll}
	\qquad\qquad
	$\abort \auxrel{\abortLab} \inact$ 
	\ \  \rulelabel{P-Abt}
	\\[.1cm]
	\hline
	\end{tabular}
	\caption{\Cherry\ semantics: auxiliary labelled relation.}
	\label{fig:semantics_cherry_pi_aux}
\end{figure*}

\begin{figure}
\modif{
\begin{align*}	
& 	
C_1 \mid C_2 \equiv C_2 \mid C_1 \qquad\
   	(C_1 \mid C_2) \mid C_3 \equiv C_1 \mid (C_2 \mid C_3)
 \qquad \recAct{X}{P} \equiv P[\recAct{X}{P} / X]&
\end{align*}
\caption{Structural congruence for \cherry}
\label{fig:struct}
}
\end{figure}

The operational rules defining the reduction relation $\fwbwred$
are reported in Fig.~\ref{fig:semantics_cherry_pi}. 
We comment on salient points.
Once a session is created 
\modifRev{(via a synchronisation along a shared channel $a$),}
its initiating collaboration is stored in the 
session construct (rule \rulelabel{F-Con}); 
\modifRev{note how a fresh session channel $s$ is generated as the result 
of the interaction.}
%
%
\modifRev{Communication, branching selection and internal conditional choice proceed as usual, without affecting logs.
Specifically, communication takes place when an output and an input action synchronise along a session channel (rule \rulelabel{F-Com}); 
the delivery of the value $v$ (resulting from the evaluation of the expression argument of the output action, rule \rulelabel{P-Snd}) is expressed by the application of a substitution $\subst{v}{x}$ to the continuation of the receiving process. 
Branching selection results from the synchronisation on a label $l$ (rule \rulelabel{F-Lab});
note that only one of the branches is selected while the remaining ones are discarded (rule \rulelabel{P-Brn}).
Internal conditional choice behaves in a standard way (rule \rulelabel{F-If}); as usual, the choice depends by 
the positive (rule \rulelabel{P-IfT}) or negative (rule \rulelabel{P-IfF}) evaluation of the boolean expression
argument of the conditional choice construct. 
}
A commit action updates the checkpoint of a session, by replacing the processes stored in the logs of the two involved parties (rule \rulelabel{F-Cmt}). 
Notably, this form of commit is asynchronous as it does not require the passive participant 
to explicitly synchronise with the active participant by means of a primitive for accepting the commit. 
On the other hand, under the hood, a low-level implementation of this mechanism would synchronously  
update the logs of the involved parties.
Conversely, a rollback action restores the processes in the two logs (rule \rulelabel{B-Rll}). 
The abort action (rule \rulelabel{B-Abt}), instead, kills the session and restores the collaboration stored in the session 
construct formed by the two initiators that have started the session; this allows the initiators to be involved in new sessions.
The other rules simply extend the standard parallel and
restriction rules to forward and backward relations. 

\begin{figure*}[t]
	\centering
	\small
	\begin{tabular}{@{}l@{}}
        $\requestAct{a}{x_1}{P_1} \mid \acceptAct{a}{x_2}{P_2}$
	$\fwred$
	$\singleSession{s}{(\requestAct{a}{x_1}{P_1} \mid \acceptAct{a}{x_2}{P_2})}$
	\hspace*{3cm}
	\rulelabel{F-Con}
	\\ 
	\hspace*{3.4cm}
        $(\loggedShort{\ce{s}}{P_1\subst{\ce{s}}{x_1} }	P_1\subst{\ce{s}}{x_1}
        \mid 
	\loggedShort{s}{P_2\subst{s}{x_2}}	P_2\subst{s}{x_2} )\ $
	\\[.3cm]	
	$
	\infer[$\rulelabel{F-Com}$]{
	\loggedShort{\ce{\genSession}}{Q_1} P_1 
	\mid 
	\loggedShort{\genSession}{Q_2} P_2
	\ \fwred\ 
	\loggedShort{\ce{\genSession}}{Q_1} P_1' 
	\mid 
	\loggedShort{\genSession}{Q_2} P_2'\subst{v}{x}	
	}
	{
	P_1 \auxrel{\send{\ce{k}}{v}} P_1'
	\qquad
	P_2 \auxrel{\receive{k}{x}} P_2'	
	}
	$
	\hspace{1.7cm}
		$
	\infer[$\rulelabel{F-Par}$]{C_1 \!\mid\! C_2 \ \fwred\ C_1' \!\mid\! C_2}
	{C_1 \ \fwred\ C_1'}
	$		
       	\\[.3cm]	
	$
	\infer[$\rulelabel{F-Lab}$]{
	\loggedShort{\ce{\genSession}}{Q_1} P_1 
	\mid 
	\loggedShort{\genSession}{Q_2} P_2
	\ \fwred\ 
	\loggedShort{\ce{\genSession}}{Q_1} P_1' 
	\mid 
	\loggedShort{\genSession}{Q_2} P_2'
	}
	{
	P_1 \auxrel{\select{\ce{k}}{l}} P_1'
	\qquad
	P_2 \auxrel{\branching{k}l} P_2'	
	}
	$
	\hspace{.9cm}
	$
	\infer[$\rulelabel{F-Res}$]{\singleSession{s}{C_1}\,C_2 \ \fwred\ \singleSession{s}{C_1}\,C_2'}
	{C_2 \ \fwred\ C_2' }
	$
	\\[.3cm]		
	$
	\infer[$\rulelabel{F-Cmt}$]{
	\logged{\ce{\genSession}}{Q_1} P_1 
	\ \mid \ 
	\logged{\genSession}{Q_2} P_2
	\ \fwred\ 
	\logged{\ce{\genSession}}{P_1'} P_1' 
	\ \mid \ 
	\logged{\genSession}{P_2} P_2
	}
	{
	P_1 \auxrel{\commitLab} P_1'
	}
	$
	\hspace{1.6cm}
		$
	\infer[$\rulelabel{F-If}$]{
	\logged{\genSession}{Q} P
	\ \fwred\ 
	\logged{\genSession}{Q} P' 
	}
	{
	P \auxrel{\ite} P'
	}
	$		
	\\[.4cm]
	$
	\infer[$\rulelabel{B-Rll}$]{
	\logged{\ce{\genSession}}{Q_1} P_1 
	\ \mid \ 
	\logged{\genSession}{Q_2} P_2
	\ \bwred\ 
	\logged{\ce{\genSession}}{Q_1} Q_1
	\ \mid \ 
	\logged{\genSession}{Q_2} Q_2
	}
	{
	P_1 \auxrel{\rollLab} P_1'
	}
	$	
	\hspace{1.2cm}
	$
	\infer[$\rulelabel{B-Par}$]{C_1 \!\mid\! C_2 \ \bwred\ C_1' \!\mid\! C_2}
	{C_1 \ \bwred\ C_1'}
	$
	\\[.4cm]		
	$
	\infer[$\rulelabel{B-Abt}$]{
	\singleSession{s}{C}(
	\logged{\ce{\genSession}}{Q_1} P_1 
	\ \mid \ 
	\logged{\genSession}{Q_2} P_2)
	\ \bwred\ C
	}
	{
	P_1 \auxrel{\abortLab} P_1'
	}
	$	
	\hspace{1.4cm}
	$
	\infer[$\rulelabel{B-Res}$]{\singleSession{s}{C_1}C_2 \ \bwred\ \singleSession{s}{C_1}C_2'}
	{
	C_2 \ \bwred\ C_2'
	}
	$
	\\[.1cm]
	\hline
	\end{tabular}	
	\caption{\Cherry\ semantics: forward and backward reduction relations. }
	\label{fig:semantics_cherry_pi}
\end{figure*}

\begin{exa}\label{example_semantics}
Consider the first \cherry\ specification of the VOD scenario given in Ex.~\ref{example_syntax}.
In the initial state $C_{\texttt{US}}$ of the collaboration, \texttt{U} and \texttt{S} can synchronise in order to 
initialise the session, thus evolving to
$
C_{\texttt{US}}^1 \ = \ 	\singleSession{s}{C_{\texttt{US}}}\ 
	(\logged{\ce{s}}{P_\texttt{U}\subst{\ce{s}}{x} } P_\texttt{U}\subst{\ce{s}}{x} \ \mid\ 
  	 \logged{s}{P_\texttt{S}\subst{s}{y}}	P_\texttt{S}\subst{s}{y} )
$.
 
Let us consider now a possible run of the session. After three reduction steps, \texttt{U} executes the commit action, obtaining the following 
runtime term:
$$
\begin{array}{@{\!\!}r@{\, }c@{\, }l}
C_{\texttt{US}}^2 & = &	\singleSession{s}{C_{\texttt{US}}}\ (\logged{}{P_\texttt{U}'} P_\texttt{U}' \
\mid\ \logged{}{P_\texttt{S}'} P_\texttt{S}' )
\\[.2cm]
P_\texttt{U}' & = & 
\receiveAct{\ce{s}}{x_{\textit{\textsf{meta}}}:\strType}{\,}
\ifPi\, (f_{\textit{\textsf{eval}}}(v_{\textit{\textsf{price}}},x_{\textit{\textsf{meta}}}))\, \thenPi \ldots
\ \
P_\texttt{S}' = 
\sendAct{s}{f_{\textit{\textsf{meta}}}(v_{\textit{\textsf{req}}})}{\,}
\branchAct{y}{\ \ldots \ }
\end{array}
$$
After four further reduction steps, \texttt{U} chooses the {\small\textsf{HD}} video quality and \texttt{S} commits as well; the resulting runtime 
collaboration is as follows: 
$$
\begin{array}{r@{\ \ }c@{\ \ }l}
C_{\texttt{US}}^3 & = & 	\singleSession{s}{C_{\texttt{US}}}\ (\logged{}{P_\texttt{U}''} P_\texttt{U}'' \ \mid\ \logged{}{P_\texttt{S}''} P_\texttt{S}'' )
\\[.2cm]
P_\texttt{U}'' & = & 
\receiveAct{\ce{s}}{x_{\textit{\textsf{testHD}}}:\strType}{\,}
\ifPi\ (f_{\textit{\textsf{HD}}}(x_{\textit{\textsf{testHD}}}))\ \thenPi\ \receiveAct{\ce{s}}{x_{\textit{\textsf{videoHD}}}:\strType}{\inact}\ \elsePi\ \roll
\\[.2cm]
P_\texttt{S}'' & = & 
\sendAct{s}{f_{\textit{\textsf{testHD}}}(v_{\textit{\textsf{req}}})}{\,}
\sendAct{s}{f_{\textit{\textsf{videoHD}}}(v_{\textit{\textsf{req}}})}{\,}
\inact
\end{array}
$$
In the next reductions, \texttt{U} evaluates the test video and decides to revert the session execution,
resulting in 
$
C_{\texttt{US}}^4 \ = \ 	\singleSession{s}{C_{\texttt{US}}}\ (\logged{}{P_\texttt{U}''} \roll \ \mid\ \logged{}{P_\texttt{S}''} 
\sendAct{s}{f_{\textit{\textsf{videoHD}}}(v_{\textit{\textsf{req}}})}{\,}\inact )
$.
The execution of the roll action restores the checkpoints $P_\texttt{U}''$ and $P_\texttt{S}''$, that is 
\mbox{$C_{\texttt{US}}^4 \ \bwred\ C_{\texttt{US}}^3$}. 
After the rollback, \texttt{U} is not able to change 
the video quality as her own commit point would have permitted; in fact, it holds
\mbox{$C_{\texttt{US}}^4 \ \bwred\!\!\!\!\!\!\!/\ \ \ \ C_{\texttt{US}}^2$}.
\end{exa}

\modifRev{It is worth noticing that our approach does not guarantee the avoidance 
of infinite loops due to taking the same decisions after rollbacks. However, \cherry\
allows the programmer to specify appropriate conditions to exit from these loops. More specifically, when a  checkpoint is restored, the state of the interaction protocol is reverted at the committed configuration and, hence, the variables substituted by the undone interactions (see rule [F-COM] in Fig.~\ref{fig:semantics_cherry_pi}) are restored as well. Anyway, internal choice decisions result from the evaluation of expressions that predicate not only on these (protocol) variables, but also on information concerning the system state and decisions taken by external actors (e.g., humans) that is not subject to the reversibility effect. Relying on this kind of information is possible to make a decision to exit from a loop. These decisions are abstracted in our example scenarios by means of relations whose definitions are left unspecified.
For example, in the VOD scenario, the selection of the video quality (HD vs. SD) taken by the user depends on video's price and metadata (stored in the corresponding protocol variables) but also on the budget and the personal opinion of the user herself; on the whole, the selection decision is abstracted by means of the relation $f_{\textit{\textsf{eval}}}()$, which will return a different result after the rollback even if invoked with the same input data. In fact, despite the interaction protocol has been reverted, the user remembers that the HD video quality was not a good choice. 
There are different ways to keep track of information of a reverted computation.
One could use alternatives \cite{LMSS13} or an external oracle \cite{Vassor21}. 
The management of the information outside the one stored in the protocol variables is out of the scope of this work; we leave this for future investigation.}
\section{Rollback safety}
\label{sec:rollback_safety}

The operational semantics of \cherry\ provides a description of the functioning of the 
primitives for programming the checkpoint-based rollback recovery in a session-based language. 
However, as shown in Ex.~\ref{example_semantics}, it does not guarantee high-level 
properties about the safe execution of the rollback. 
To prevent such undesired rollbacks, we propose the use of \emph{compliance checking}, 
to be performed at design time. 
This check is not done 
on the full system specification, but only at the level of session types. 
\modifRev{In this way, we abstract as sorts the values exchanged via 
communication actions at process level. This not only permits formulating the compliance 
check more simply but, even more importantly, it ensures that the semantics produces finite LTSs, 
thus making the compliance check decidable.}

\subsection{Session types and typing}
The syntax of the \cherry\ \emph{session types} $\typeT$ is defined 
in Fig.~\ref{fig:typeSyntax_cherry}. 
\begin{figure}[t]
	\centering
	\small
	\modif{
	\begin{tabular}{@{}r@{\ }c@{\ }l@{\quad}l@{}}
	\\[-.7cm]
	$S$ & ::= & $\boolType$ \quad $\mid$ \quad $\intType$ \quad $\mid$ \quad $\strType$  
	& \textbf{Sorts}	
	\\[.2cm]	
	$\typeT$ & ::= & $\outType{S}.\typeT$ \ \ $\mid$ \ \
	$\inpType{S}.\typeT$ \ \  $\mid$ \ \ 
	$\selType{l}.\typeT$ 	 \ \  $\mid$ \ \ 
	$\branchType{\branch{l_1}{\typeT_1}, \ldots, \branch{l_n}{\typeT_n}}$	
	&  \textbf{Session types}
	\\[.05cm]
	& $\mid$ &  
	$\typeT_1 \choiceType \typeT_2$ \ \ $\mid$ \ \ 
	$t$ \ \ $\mid$ \ \ $\recType{t}.\typeT$ \ \ $\mid$ \ \ 
	$\inactType$ \ \ $\mid$ \ \ 
        \added{$\errType$} \ \ $\mid$ \ \ 
	$\commitType.\typeT$ \ \ $\mid$ \ \ 
	$\rollType$ \ \ $\mid$ \ \ 
	$\abortType$
	\\[.1cm]
	\hline
	\end{tabular}}
	\caption{\Cherry\ type syntax. 
	}
	\label{fig:typeSyntax_cherry}
\end{figure}
Type $\outType{S}.\typeT$ represents the behaviour of first outputting a value of sort $S$
(i.e., $\boolType$, $\intType$ or $\strType$), 
then performing the actions prescribed by type $\typeT$. Type $\inpType{S}.\typeT$ is the dual one, where a value is received instead of sent. 
Type $\selType{l}.\typeT$ represents the behaviour that selects the label $l$ and then behaves as $\typeT$. 
Type $\branchType{\branch{l_1}{\typeT_1}, \ldots, \branch{l_n}{\typeT_n}}$ describes a branching 
behaviour: it waits for one of the $n$ options to be selected, and behaves as type $\typeT_i$ if the $i$-th label is selected 
(external choice).
Type $\typeT_1 \choiceType \typeT_2$ behaves as either $\typeT_1$ or 
$\typeT_2$ (internal choice).
Type $\recType{t}.\typeT$ represents a recursive behaviour
\modif{that starts by doing $\typeT$ and, when variable $t$ is encountered, recurs to $\typeT$ again.} 
Types $\inactType$ and $\errType$ represent inaction and faulty termination, respectively.
Type $\commitType.\typeT$ represents a commit action followed by the actions prescribed by type 
$\typeT$. 
Finally, types $\rollType$ and $\abortType$ represent rollback and abort actions. 
\modifRev{Notably, \cherry\ session types are defined as in \cite{YoshidaV07}, but session delegation (i.e., channel-passing) is not supported while there are additional types corresponding to the reversibility actions. Concerning choice, the type discipline, as in \cite{YoshidaV07}, supports internal choice (corresponding to the if-then-else construct at process level) and external choice (corresponding to a label synchronization via interaction between selection and branching constructs at process level). Therefore, no other forms of choice are supported and, differently from other works (e.g., \cite{BarbaneraDd14,BarbaneraDLd16,BarbaneraLd17}), choices are not implicit checkpoints and rollback is explicitly triggered. }

The \cherry\ type system does not perform compliance checks, but 
only infers the types of collaboration participants, which will be then checked together 
according to the compliance relation. 
\emph{Typing judgements} are of the form \mbox{$C \hasType \sessions$},
where $\sessions$, called \emph{type associations}, is a set of session type 
associations of the form $\reqOrApp{a}:\typeT$, where $\reqOrApp{a}$
stands for either $\ce{a}$ or $a$.
Intuitively, 
\mbox{$C \hasType \sessions$} indicates that from 
the collaboration $C$ the type associations in $\sessions$ are inferred. 
The definition of the type system for these judgements relies on auxiliary 
typing judgements for processes, of the form $\basis;\sorting \judge P \hasType \typing$,
where $\basis$, $\sorting$ and $\typing$, called \emph{basis}, \emph{sorting} and \emph{typing}
respectively,  are finite partial maps 
from process variables to type variables,
from variables to sorts, and 
from variables to types, respectively.
Updates of basis and sorting are denoted, respectively, 
by $\basis\comp X:t$ and $\sorting\comp y:S$, where 
$X \notin dom(\basis)$, $t \notin cod(\basis)$ and $y \notin dom(\sorting)$.
The judgement \mbox{$\basis;\sorting \judge P \hasType \typing$ }
stands for ``under the environment $\basis;\sorting$, process $P$ has typing $\typing$''. 
In its own turn, the typing of processes relies on auxiliary judgments for
expressions, of the form \mbox{$\sorting \judge e \hasType S$}.
\modif{The axioms and rules defining the typing system for \cherry\ collaborations, processes 
and expressions are given in Fig.~\ref{fig:typingSystem_coll},~\ref{fig:typingSystem_proc}, and~\ref{fig:typingSystem_exp}, 
respectively.} 
%
The type system is defined only for initial collaborations, i.e. for terms generated by the grammar in Fig.~\ref{fig:syntax_cherry_pi}. Other runtime collaborations are not considered here, as no check will be performed at runtime.
We comment on salient points. Typing rules at collaboration level simply collect the type associations 
of session initiators in the collaboration. Rules at process level instead determine the session 
type corresponding to each process, by mapping each process operator to the corresponding type operator. 
Data and expression used in communication actions are abstracted as
sorts, and
a conditional choice is rendered as an internal non-deterministic choice. 
\modif{Typing rules for expressions are standard.}

\begin{figure}[!t]
	\centering
	\small
	\begin{tabular}{c}
	$
	\infer[$\!\!\rulelabel{T-Req}$]{\requestAct{a}{x}{P} \hasType \{\ce{a}:\typeT\}}
	{\emptyset;\emptyset \!\!\judge\!\!\! P \!\!\hasType\! x:\typeT}
	$		
	\quad\
	$
	\infer[$\!\!\rulelabel{T-Acc}$]{\acceptAct{a}{x}{P} \hasType \{a:\typeT\}}
	{\emptyset;\emptyset \!\!\judge\!\!\! P \!\!\hasType\! x:\typeT}
	$
	\quad\
	$
	\infer[$\!\!\rulelabel{T-Par}$]{C_1\mid C_2 \hasType \sessions_1 \cup \sessions_2}
	{C_1 \hasType \sessions_1 & C_2 \hasType \sessions_2}
	$
        \\[.1cm]
	\hline
	\end{tabular}
	\caption{Typing system for \cherry\ collaborations.}
	\label{fig:typingSystem_coll}
\end{figure}

\begin{figure}[!h]
	\centering
	\small
	\begin{tabular}{@{}c@{}}
		$
	\infer[$\rulelabel{T-Snd}$]{\basis;\sorting \judge \sendAct{x}{e}{P} \hasType x:\outType{S}.\typeT}
	{\sorting \judge e \hasType S && \basis;\sorting \judge P \hasType x:\typeT}
	$		
	\qquad
	$
	\infer[$\rulelabel{T-Rcv}$]{\basis;\sorting \judge \receiveAct{x}{y:S}{P} \hasType x:\inpType{S}.\typeT}
	{\basis;\sorting\comp y:S \judge P \hasType x:\typeT}
	$	
	\\[.4cm]	
	\raisebox{.2cm}{$\basis;\sorting\! \judge\! \inact\! \hasType\! x:\inactType$\ \ \rulelabel{T-Inact}}
	\qquad\qquad\qquad
	\modif{
	$
	\infer[$\rulelabel{T-Sel}$]{\basis;\sorting \judge \selectAct{x}{l}{P} \hasType 
	x:\selType{l}.\typeT}
	{\basis;\sorting  \judge P \hasType x:\typeT}
	$      
	}  
	\\[.3cm]			
	\modif{
	$
	\infer[$\rulelabel{T-Br}$]{\basis;\sorting \judge \branchAct{x}{\branch{l_1}{P_1} \branchSep \ldots \branchSep \branch{l_n}{P_n}} \hasType 
	x:\branchType{\branch{l_1}{\typeT_1}, \ldots, \branch{l_n}{\typeT_n}}}
	{\basis;\sorting  \judge P_1 \hasType x:\typeT_1 && \ldots 
	 && \basis;\sorting  \judge P_n \hasType x:\typeT_n}
	$}
	\\[.5cm]	
       	$
	\infer[$\rulelabel{T-If}$]{\basis;\sorting \judge \ifthenelseAct{e}{P_1}{P_2} \hasType 
	x:\typeT_1 \choiceType \typeT_2}
	{\sorting \! \judge\! e \!\hasType\! \boolType &&
	\basis;\sorting  \judge P_1 \hasType x:\typeT_1 && \basis;\sorting  \judge P_2 \hasType x:\typeT_2}
	$
	\\[.4cm]	
	\raisebox{.2cm}{$\basis\comp X:t;\sorting \judge X \!\hasType\! t$\ \ \rulelabel{T-Pvar}}			
	\qquad
	\infer[$\rulelabel{T-Rec}$]{\basis;\sorting \judge \recAct{X}{P} \hasType \recAct{t}{\typeT}}
	{\basis\comp X:t;\sorting  \judge P \hasType \typeT}	
	\\[.4cm]
	$
	\infer[$\rulelabel{T-Cmt}$]{\basis;\sorting \judge \commitAct{P}  \hasType 
	x:\commitType.\typeT}
	{
	\basis;\sorting  \judge P \hasType x:\typeT 
	}
	$
	\\[.4cm]
	$\basis;\sorting \judge \roll \hasType x:\rollType$\ \ \rulelabel{T-Rll}	
	\qquad
	$\basis;\sorting \judge \abort \hasType x:\abortType$\ \ 
	\rulelabel{T-Abt}
	\\[.1cm]	
	\hline
	\end{tabular}
	\caption{Typing system for \cherry\ processes.}
	\label{fig:typingSystem_proc}
\end{figure}

\begin{figure}[t]
	\centering
	\small
	\modif{
	\begin{tabular}{@{}c@{}}
	$\sorting \judge \ctrue \hasType \boolType$\ \ \rulelabel{T-Bool$_{\mathit{tt}}$}
	\qquad
	$\sorting \judge \cfalse \hasType \boolType$\ \ \rulelabel{T-Bool$_{\mathit{ff}}$}
	\\[.4cm]	
	$\sorting\comp x:S \judge x \hasType S$\ \ \rulelabel{T-Var}	
	\qquad
	$\sorting \judge 1 \hasType \intType$\ \ \rulelabel{T-Int}	
	\qquad
	$\sorting \judge\!\!``a" \hasType \strType$\ \ \rulelabel{T-Str}	
	\\[.4cm]
	$
	\infer[$\!\rulelabel{T-Sum}$]{\sorting \judge +(e_1,e_2) \hasType \intType}
	{\sorting \judge e_1 \hasType \intType &  \sorting \judge e_2 \hasType \intType}
	\qquad
	\infer[$\!\rulelabel{T-And}$]{\sorting \judge \wedge(e_1,e_2) \hasType \boolType}
	{\sorting \judge e_1 \hasType \boolType  & \sorting \judge e_2 \hasType \boolType}
	$		
        \\[.2cm]
	\hline
	\end{tabular}}
	\caption{Typing system for \cherry\ expressions (excerpt of
          rules).}
	\label{fig:typingSystem_exp}
\end{figure}

\subsection{Compliance checking}
To check compliance between pairs of session parties, we consider
\emph{type configurations} of the form 
$
\initConf{\typeT}{\typeT'}\conf{\checkpointType{\typeT}_1}{\typeT_2} \confcomp \conf{\checkpointType{\typeT}_3}{\typeT_4}
$,  
consisting in a pair $(\typeT,\typeT')$ of session types, corresponding to the types of the parties at the initiation of the session, and in the parallel composition of two pairs $\conf{\checkpointType{\typeT}_c}{\typeT}$, where $\typeT$ is the session type of a party and $\checkpointType{\typeT}_c$ is the type of the party's checkpoint. We use $\checkpointType{\typeT}$ to denote either a type $\typeT$, representing a checkpoint committed by the party, or $\imposed{\typeT}$, representing a checkpoint imposed by the other party.
%
The semantics of type configurations, necessary for the definition of the compliance relation, is given in Fig.~\ref{fig:typeSemantics_ext}, where label $\typeLab$ stands for either 
$\outType{S}$,
$\inpType{S}$,
$\selTypeLabel{l}$,
$\branchTypeLabel{l}$,
$\tau$,
$\commitType$,
$\rollType$,
or $\abortType$.
We comment on the relevant rules.
In case of a commit action, the checkpoints of both parties are updated, 
and the one 
of the passive party (i.e., the party that has not performed the commit) 
is marked as `imposed' (rule \rulelabel{TS-Cmt$_1$}). 
However, if the passive party did not perform any action from its 
current checkpoint, this checkpoint is not overwritten by the active party (rule \rulelabel{TS-Cmt$_2$}),
as discussed in Sec.~\ref{sec:motivation} (Fig.~\ref{fig:scenario}(c)). 
In case of a roll action (rule \rulelabel{TS-Rll$_1$}), the reduction step is performed 
only if the active party (i.e., the party that has performed the rollback 
action) has a non-imposed checkpoint; 
\modif{otherwise, the configuration cannot proceed with the rollback and reduces to 
an erroneous configuration (rule \rulelabel{TS-Rll$_2$}).}
Finally, in case of abort (rule \rulelabel{TS-Abt$_1$}), the configuration goes back to the initial state; this 
allows the type computation to proceed, in order not to affect the compliance check
between the two parties. 

\begin{figure}[!t]
	\centering
	\small
	\begin{tabular}{@{}l@{\hspace*{.8cm}}l@{\hspace*{.6cm}}r@{}}
	\modif{
	$\outType{S}.\typeT \typeTrans{\outType{S}} \typeT$\ \ \rulelabel{TS-Snd}	
	}
	&
	\modif{
	$\inpType{S}.\typeT \typeTrans{\inpType{S}} \typeT$\ \ \rulelabel{TS-Rcv}  
	}
	&
	\modif{
	$\selType{l}.\typeT \typeTrans{\selTypeLabel{l}} \typeT$\ \ \rulelabel{TS-Sel}
	}
	\\[.4cm]
	\multicolumn{2}{@{}l@{}}{\modif{
	\raisebox{.8cm}{$\!\!
	\branchType{\branch{l_1\!\!}{\!\typeT_1}, \ldots, \branch{l_n\!\!}{\!\typeT_n}}
	\typeTrans{\branchTypeLabel{l_i}} \typeT_i$ \ \ $(1 \!\leq\! i \!\leq\! n)$\ \ \rulelabel{TS-Br}
	}}
	}
	&
	\modif{
	$
	\infer[$\rulelabel{TS-rec}$]
	{\recType{t}.\typeT \typeTrans{\typeLab} \typeT'}
	{\typeT\subst{\recType{t}.\typeT}{t} \typeTrans{\typeLab} \typeT'}
	$
	}
\\[.3cm]
	\modif{
	$
	\typeT_1 \choiceType \typeT_2 \typeTrans{\tau} \typeT_1$\ \ \rulelabel{TS-If$_1$}
	}
	&
	\modif{
	$\typeT_1 \choiceType \typeT_2 \typeTrans{\tau} \typeT_2$\ \ \rulelabel{TS-If$_2$}
	}
	&
	\\[.4cm]
	$\commitType.\typeT \typeTrans{\commitLab} \typeT$\ \ \rulelabel{TS-Cmt}
	&
	$\rollType \typeTrans{\rollLab} \inactType$\ \ \rulelabel{TS-Rll}
	&
	$\abortType \typeTrans{\abortLab} \inactType$\ \ \rulelabel{TS-Abt}
	\\[.3cm]
	\end{tabular}
	\begin{tabular}{@{}c@{}}
	$
	\infer[$\rulelabel{TS-Tau}$]
	{\initConf{\typeT}{\typeT'}\conf{\checkpointType{\typeU}_1}{\typeT_1} \confcomp \conf{\checkpointType{\typeU}_2}{\typeT_2}
	\typered
	\initConf{\typeT}{\typeT'}\conf{\checkpointType{\typeU}_1}{\typeT_1'} \confcomp \conf{\checkpointType{\typeU}_2}{\typeT_2}}
	{\typeT_1 \typeTrans{\tau} \typeT_1'}
	$	
	\\[.3cm]
	$
	\infer[$\rulelabel{TS-Com}$]
	{\initConf{\typeT}{\typeT'}\conf{\checkpointType{\typeU}_1}{\typeT_1} \confcomp \conf{\checkpointType{\typeU}_2}{\typeT_2}
	\typered
	\initConf{\typeT}{\typeT'}\conf{\checkpointType{\typeU}_1}{\typeT_1'} \confcomp \conf{\checkpointType{\typeU}_2}{\typeT_2'}}
	{\typeT_1 \typeTrans{\!\outType{S}\!} \typeT_1'
	& &
	 \typeT_2 \typeTrans{\!\!\inpType{S}\!} \typeT_2'}	
	$	
	\\[.3cm]
	\modif{
	$
	\infer[$\rulelabel{TS-Lab}$]
	{\initConf{\typeT}{\typeT'}\conf{\checkpointType{\typeU}_1}{\typeT_1} \confcomp \conf{\checkpointType{\typeU}_2}{\typeT_2}
	\typered
	\initConf{\typeT}{\typeT'}\conf{\checkpointType{\typeU}_1}{\typeT_1'} \confcomp \conf{\checkpointType{\typeU}_2}{\typeT_2'}}
	{\typeT_1 \typeTrans{\selTypeLabel{l}} \typeT_1'
	& & 
	 \typeT_2 \typeTrans{\branchTypeLabel{l}} \typeT_2'}
	$	
	}
	\\[.4cm]
	$
	\infer[$\rulelabel{TS-Cmt$_1$}$]
	{\initConf{\typeT}{\typeT'}\conf{\checkpointType{\typeU}_1}{\typeT_1} \confcomp \conf{\checkpointType{\typeU}_2}{\typeT_2}
	\typered
	\initConf{\typeT}{\typeT'}\conf{\typeT_1'}{\typeT_1'} \confcomp \conf{\imposed{\typeT_2}}{\typeT_2}}
	{\typeT_1 \typeTrans{\commitLab} \typeT_1' & & \checkpointType{\typeU}_2 \neq \typeT_2}
	$	
	\\[.3cm]
	$
	\infer[$\rulelabel{TS-Cmt$_2$}$]
	{\initConf{\typeT}{\typeT'}\conf{\checkpointType{\typeU}_1}{\typeT_1} \confcomp \conf{\checkpointType{\typeU}_2}{\typeT_2}
	\typered
	\initConf{\typeT}{\typeT'}\conf{\typeT_1'}{\typeT_1'} \confcomp \conf{\checkpointType{\typeU}_2}{\typeT_2}}
	{\typeT_1 \typeTrans{\commitLab} \typeT_1' & & \checkpointType{\typeU}_2 = \typeT_2}
	$	
	\\[.3cm]
	$
	\infer[$\rulelabel{TS-Rll$_1$}$]
	{\initConf{\typeT}{\typeT'}\conf{\typeU_1}{\typeT_1} \confcomp \conf{\checkpointType{\typeU}_2}{\typeT_2}
	\typered
	\initConf{\typeT}{\typeT'}\conf{\typeU_1}{\typeU_1} \confcomp \conf{\checkpointType{\typeU}_2}{\typeU_2}}
	{\typeT_1 \typeTrans{\rollLab} \typeT_1'}
	$
        \\[.3cm]
	\added{$
	\infer[$\rulelabel{TS-Rll$_2$}$]
	{\initConf{\typeT}{\typeT'}\conf{\imposed{\typeU_1}}{\typeT_1} \confcomp \conf{\checkpointType{\typeU}_2}{\typeT_2}
	\typered
	\initConf{\typeT}{\typeT'}\conf{\imposed{\typeU_1}}{\errType} \confcomp \conf{\checkpointType{\typeU}_2}{\errType}}
	{\typeT_1 \typeTrans{\rollLab} \typeT_1'}
	$}
	\\[.3cm]
	$
	\infer[$\rulelabel{TS-Abt$_1$}$]
	{\initConf{\typeT}{\typeT'}\conf{\checkpointType{\typeU}_1}{\typeT_1} \confcomp \conf{\checkpointType{\typeU}_2}{\typeT_2}
	\typered
	\initConf{\typeT}{\typeT'}\conf{\typeT}{\typeT} \confcomp \conf{\typeT'}{\typeT'}}
	{\typeT_1 \typeTrans{\abortLab} \typeT_1'}
	$	
	\\[-.3cm]
	\hspace*{12.1cm}\\
	\hline
	\end{tabular}
	\caption{Semantics of types and type configurations 
	(symmetric rules for configurations are omitted).
	}
	\label{fig:typeSemantics_ext}
\end{figure}

On top of the above type semantics, 
we define the compliance relation, inspired by the relation in
\cite{BarbaneraDLd16}, and prove its decidability.

\begin{defi}[Compliance]\label{def:compliance}
Relation $\compliant$ on configurations is defined as follows:
$\initConf{\typeT}{\typeT'}\conf{\checkpointType{\typeU_1}}{\typeT_1} \compliant \conf{\checkpointType{\typeU_2}}{\typeT_2}$
holds if 
\added{for all}
$\typeU_1'$, $\typeT_1'$, $\typeU_2'$, $\typeT_2'$ such that
$\initConf{\typeT}{\typeT'}\conf{\checkpointType{\typeU_1}}{\typeT_1} \confcomp \conf{\checkpointType{\typeU_2}}{\typeT_2}
\typered^*$
$\initConf{\typeT}{\typeT'}\conf{\checkpointType{\typeU_1'}}{\typeT_1'} \confcomp \conf{\checkpointType{\typeU_2'}}{\typeT_2'}
\typered\!\!\!\!\!\!\!\!/\ \ \  $
we have that $\typeT_1'=\typeT_2'=\inactType$.
Two types $\typeT_1$ and $\typeT_2$ are \emph{compliant}, written 
$\typeT_1 \compliant \typeT_2$, if  
$\initConf{\typeT_1}{\typeT_2}\conf{\typeT_1}{\typeT_1} \compliant \conf{\typeT_2}{\typeT_2}$.
\end{defi}

\begin{thm}\label{th:decidability}
Let $\typeT_1$ and $\typeT_2$ be two session types, checking if $\typeT_1 \compliant \typeT_2$ holds
is decidable.
\end{thm}
\modif{
%
%
\begin{proof} 
By Def.~\ref{def:compliance}, checking $\typeT_1 \compliant \typeT_2$ consists in
checking that types $\typeT_1'$ and $\typeT_2'$ of each configuration
$\initConf{\typeT_1}{\typeT_2}\conf{\checkpointType{\typeU_1'}}{\typeT_1'} 
\confcomp \conf{\checkpointType{\typeU_2'}}{\typeT_2'}$
such that 
$\initConf{\typeT_1}{\typeT_2}\conf{\typeT_1}{\typeT_1} \confcomp \conf{\typeT_2}{\typeT_2}
\typered^*
\initConf{\typeT_1}{\typeT_2}\conf{\checkpointType{\typeU_1'}}{\typeT_1'} \confcomp
\conf{\checkpointType{\typeU_2'}}{\typeT_2'}
\typered\!\!\!\!\!\!\!\!/\ \ \  $
(i.e., type configurations that are reachable from the initial one and that cannot further evolve)
are $\inactType$ types.
Thus, to prove that the compliance check is decidable we have to show that the number of 
these reachable configurations is finite. 
Let us consider the transition system $TS=\langle \mathcal{S}, \mathcal{R} \rangle$ associated to the type configuration 
$t= \initConf{\typeT_1}{\typeT_2}\conf{\typeT_1}{\typeT_1} \confcomp \conf{\typeT_2}{\typeT_2}$
by the  reduction semantics of types (Fig.~\ref{fig:typeSemantics_ext}):
the set $\mathcal{S}$ of states  corresponds to the set of type configurations reachable from $t$,
i.e. $\mathcal{S}=\{\, t' \ \mid \  t \typered^* t'\,\}$, 
while the set $\mathcal{R}$ of system transitions corresponds to set of the type reductions 
involving configurations in $\mathcal{S}$, i.e.
$\mathcal{R}=\{ (t',t'') \in \mathcal{S} \times \mathcal{S}  \ | \ t' \typered t''\}$.
Hence, checking $\typeT_1 \compliant \typeT_2$ boils down to check the type configurations
corresponding to the leaves (i.e., states without outgoing transitions) of $TS$. 
Specifically, given a leaf of $TS$ corresponding to 
$\initConf{\typeT_1}{\typeT_2}\conf{\checkpointType{\typeV_1}}{\typeV_2} \confcomp
\conf{\checkpointType{\typeV_3}}{\typeV_4}$, 
we have to check if $\typeV_2=\typeV_4=\inactType$.
The decidability of this check therefore depends on the finiteness of $TS$.
This result is ensured by the fact that: : \emph{(i)} backward reductions connect states of $TS$
only to previously visited states of $TS$ (Theorem~\ref{th:trace}),
and \emph{(ii)} our language of types (Fig.~\ref{fig:typeSyntax_cherry}) 
corresponds to a CCS-like process algebra without 
static operators (i.e., parallel and restriction operators) within recursion (see \cite[Sec.~7.5]{CCS}). 
\end{proof}
}

\medskip

%
The compliance relation is used to define the notion of \emph{rollback safety}.

\begin{defi}[Rollback safety]\label{def:rollback_safety}
Let $C$ be an initial collaboration, then $C$ is \emph{rollback safe} (shortened \emph{\rollSafe}) if $C \hasType \sessions$
and for all pairs $\ce{a}:\typeT_1$ and $a:\typeT_2$ in $\sessions$ we have 
$\typeT_1 \compliant \typeT_2$.
\end{defi}


\begin{exa}\label{ex:types}
Let us consider again the VOD example. 
As expected, the first \cherry\ collaboration defined in Ex.~\ref{example_syntax}, corresponding to the 
scenario described in Fig.~\ref{fig:scenario}(b), is not rollback safe, because 
the types of the two parties are not compliant. 
Indeed, the session types $\typeT_\texttt{U}$ and $\typeT_\texttt{S}$ associated by the type system to the user and 
the service processes, respectively, are as follows:
$$
\begin{array}{r@{\ \ }c@{\ \ }l}
\typeT_\texttt{U} & = & 
\outType{\strType}.\,
\inpType{\intType}.\,
\commitType.\,
\inpType{\strType}.\,
(\, \selType{l_{\textit{\textsf{HD}}}}.\, \inpType{\strType}.\, (\, \inpType{\strType}.\,\inactType  \ \choiceType\ \rollType\,)\\
&&\qquad\qquad\qquad\qquad \qquad\qquad
\choiceType \ 
\selType{l_{\textit{\textsf{SD}}}}.\, \inpType{\strType}.\, (\, \inpType{\strType}.\,\inactType \ \choiceType\ \abortType\,)\,)
\\[.2cm]
\typeT_\texttt{S} & = & 
\inpType{\strType}.\,
\outType{\intType}.\,
\outType{\strType}.\,
\branchType{\branch{l_{\textit{\textsf{HD}}}}{\commitType.\, \outType{\strType}.\,\outType{\strType}.\,\inactType}
\ ,\\
&& \hspace*{3.77cm} 
\branch{l_{\textit{\textsf{SD}}}}{\commitType.\, \outType{\strType}.\,\outType{\strType}.\,\inactType}}
\end{array}
$$
Thus, the resulting initial configuration is 
$\initConf{\typeT_\texttt{U}}{\typeT_\texttt{S}}\conf{\typeT_\texttt{U}}{\typeT_\texttt{U}} \confcomp \conf{\typeT_\texttt{S}}{\typeT_\texttt{S}}$, 
which can evolve to the configuration 
$
\initConf{\typeT_\texttt{U}}{\typeT_\texttt{S}}\conf{\imposed{\typeT}}{\rollType} \confcomp \conf{\typeU}{\outType{\strType}.\inactType}
$,
with 
$\typeT = \inpType{\strType}.\, (\inpType{\strType}.\,\inactType  \ \choiceType\ \rollType)$
and 
$\typeU = \outType{\strType}.\,\outType{\strType}.\,\inactType$.
\added{This configuration evolves to 
$
\initConf{\typeT_\texttt{U}}{\typeT_\texttt{S}}\conf{\imposed{\typeT}}{\errType} \confcomp \conf{\typeU}{\errType}
$, which cannot further evolve and is not in a completed state (in fact, type $\errType$ is different from $\inactType$), meaning that ${\typeT_\texttt{U}}$ and ${\typeT_\texttt{S}}$ are not compliant.}

In the scenario described in Fig.~\ref{fig:scenario}(c), instead, the type of the server 
process is as follows:
$
\typeT_\texttt{S}' = \, 
\inpType{\strType}.\,
\outType{\intType}.\,
\commitType.\,
\outType{\strType}.\,
\branchType{\branch{l_{\textit{\textsf{HD}}}}{\outType{\strType}.\,\outType{\strType}.\,\inactType}
\, ,\,
\branch{l_{\textit{\textsf{SD}}}}{\outType{\strType}.$ $\outType{\strType}.\,\inactType}}
$
and we have $\typeT_\texttt{U} \compliant \typeT_\texttt{S}'$. 
Finally, the types of the processes depicted in Fig.~\ref{fig:scenario}(d) are:
$$
\begin{array}{@{}r@{\ }c@{\  }l}
\typeT_\texttt{U}' & = & 
\outType{\strType}.\,
\inpType{\intType}.\,
\inpType{\strType}.\,
\commitType.\,
(\, \selType{l_{\textit{\textsf{HD}}}}.\, \ldots\ \choiceType \ 
\selType{l_{\textit{\textsf{SD}}}}.\, \ldots\,)
\\[.2cm]
\typeT_\texttt{S}'' & = & 
\inpType{\strType}.\,
\outType{\intType}.\,
\outType{\strType}.\,
\commitType.
\branchType{\branch{l_{\textit{\textsf{HD\!\!}}}}{\outType{\strType}.\,\outType{\strType}.\,\inactType}
 ,
\branch{l_{\textit{\textsf{SD\!\!}}}}{\outType{\strType}.\,\outType{\strType}.\,\inactType}}
\end{array}
$$
and we have $\typeT_\texttt{U}' \compliant\!\!\!\!/\ \ \typeT_\texttt{S}''$. 
Indeed, the corresponding initial configuration can evolve to the configuration 
$
\initConf{\typeT_\texttt{U}'}{\typeT_\texttt{S}''}\conf{\imposed{\selType{l_{\textit{\textsf{HD}}}}.\, \ldots}}{\rollType} 
\confcomp \conf{\branchType{\branch{l_{\textit{\textsf{HD}}}}{\ldots}, 
\branch{l_{\textit{\textsf{SD}}}}{\ldots}}}{\outType{\strType}.\inactType}
$, 
which again 
\added{evolves to a configuration that is not in a completed state.}
\end{exa}

\section{\maude\ implementation.}
\label{sec:maude}
To show the feasibility of our approach, we have implemented
the semantics of type configurations in Fig.~\ref{fig:typeSemantics_ext}
in the \maude{} framework \cite{maude2007}.
\maude{} provides an instantiation of rewriting logic \cite{meseguer_conditional_1992} 
and it has been used to implement the semantics of several formal languages
\cite{DBLP:journals/jlp/Meseguer12}.

The syntax of \cherry\ types and type configurations is specified by defining 
algebraic data types, while transitions and reductions are rendered as rewrites and, hence, 
inference rules are given in terms of (conditional) rewrite rules. 
Since \maude{} specifications are executable, we have obtained in this way an 
interpreter for \cherry\ type configurations, which permits to explore the reductions
arising from the initial configuration of two given session types. 

Our implementation consists of two \maude{} modules.
The CHERRY-TYPES-SYNTAX module provides the definition of the sorts
	that characterise the syntax of \cherry\ types, such as session types, 
	selection/branching labels, type variables and type configurations. 
	In particular, basic terms of session types are rendered as constant 
	\emph{operations} 	on the sort $\texttt{Type}$;  e.g., the roll type is defined as 
\begin{lstlisting}[frame=single,basicstyle=\scriptsize\ttfamily]
op roll : -> Type .
\end{lstlisting}
	The other syntactic operators are instead defined as operations with one or 
	more arguments; e.g., the output type takes as input a $\texttt{Sort}$ 
	and a continuation type:
\begin{lstlisting}[frame=single,basicstyle=\scriptsize\ttfamily]	
op ![_]._ : Sort Type -> Type [frozen prec 25] .
\end{lstlisting}	
        To prevent undesired rewrites inside operator arguments, 
        following the approach 
        in~\cite{CCSMaude2}, we have declared these operations 
        as $\texttt{frozen}$. The $\texttt{prec}$ attribute has been used
        to define the precedence among operators. 

The CHERRY-TYPES-SEMANTICS module provides \emph{rewrite rules}, and 
	additional operators and equations, to define the \cherry\ type semantics.
	For example, the operational rule \rulelabel{TS-Snd} 
	is rendered as follows:
\begin{lstlisting}[frame=single,basicstyle=\scriptsize\ttfamily]	
rl [TS-Snd] : ![S].T => {![S]}T .
\end{lstlisting}	
        The correspondence between the operational rule and the rewrite rule is 
        one-to-one; the only peculiarity is the fact that,  since rewrites have no 
        labels, we have made the transition label part of the resulting term.  
        Reduction rules for type configurations are instead rendered in terms of 
        conditional rewrite rules with rewrites in their conditions. 
        	For example, the \rulelabel{TS-Com} rule is rendered as:
\begin{lstlisting}[frame=single,basicstyle=\scriptsize\ttfamily]	
crl [TS-Com] : 
  init(T,T') CT1 > T1 || CT2 > T2  =>  init(T,T') CT1 > T1' || CT2 > T2'
  if T1 => {![S]}T1' /\ T2 => {?[S]}T2' .
\end{lstlisting}	
        Again, there is 
        a close correspondence between 
        the operational rule and the rewrite one. 

The compliance check between two session types can be then conveniently 
realised on top of the implementation described above by resorting to 
the \maude\ command $\texttt{search}$. This permits indeed to explore 
the state space of the configurations reachable from an initial configuration. 
Specifically, the compliance check between types $\texttt{T1}$ and $\texttt{T2}$
is rendered as follows:
\begin{lstlisting}[frame=single,basicstyle=\scriptsize\ttfamily]	
search 
  init(T1,T2) ckp(T1) > T1 || ckp(T2) > T2 
  =>! 
  init(T:Type,T':Type) CT1:CkpType > T1':Type || CT2:CkpType > T2':Type 
such that T1' =/= end or T2' =/= end .
\end{lstlisting}	
This command searches for all terminal states (\texttt{=>!}), i.e. 
states that cannot be rewritten any more (see $\typered\!\!\!\!\!\!\!\!/\ \ \ $ in 
Def.~\ref{def:compliance}), and checks if at least one of the two session types in 
the corresponding configurations ($\texttt{T1'}$ and $\texttt{T2'}$) is different from 
the $\texttt{end}$ type. 
Thus, if this search has no solution,  $\texttt{T1}$ and $\texttt{T2}$ are compliant;
otherwise, they are not compliant and a violating configuration is returned.  

\begin{exa}
Let us consider the \cherry\ types defined in Ex.~\ref{ex:types} for the 
scenario described in Fig.~\ref{fig:scenario}(b).
In our \maude\ implementation of the type syntax, the session types $\typeT_\texttt{U}$ 
and $\typeT_\texttt{S}$, and the corresponding initial type configuration, are rendered 
as follows:

\begin{lstlisting}[frame=single,basicstyle=\scriptsize\ttfamily]	
eq Tuser = ![str]. ?[int]. cmt. ?[str]. 
           ((sel['hd]. ?[str]. ((?[str]. end) (+) roll)) 
            (+) (sel['sd]. ?[str]. ((?[str]. end) (+) abt))) .

eq Tservice = ?[str]. ![int]. ![str]. 
              brn[brnEl('hd, cmt. ![str]. ![str]. end); 
                  brnEl('sd, cmt. ![str]. ![str]. end)] .
  
eq InitConfig = init(Tuser,Tservice) 
                ckp(Tuser) > Tuser || ckp(Tservice) > Tservice .
\end{lstlisting}	
where $\texttt{(+)}$ represents the internal choice operator, 
$\texttt{sel}$ the selection operator, $\texttt{brn}$ the branching operator,  
$\texttt{brnEl}$ an option offered in a branching, and $\texttt{ckp}$ a non-imposed
checkpoint.
The compliance between the two session types can be checked by 
loading the two modules of our \maude\ implementation,
and executing the following command:
\begin{lstlisting}[frame=single,basicstyle=\scriptsize\ttfamily]	
search InitConfig  
       =>! 
       init(T:Type,T':Type) CT1:CkpType > T1:Type || CT2:CkpType > T2:Type 
such that T1 =/= end or T2 =/= end .
\end{lstlisting}	
This $\texttt{search}$ command returns the following solution:
\begin{lstlisting}[frame=single,basicstyle=\scriptsize\ttfamily]	
CT1 --> ickp(?[str]. ((?[str]. end)(+)roll))
T1 --> err
CT2 --> ckp(![str]. ![str]. end)
T2 --> err
\end{lstlisting}
\added{As  explained in Ex.~\ref{ex:types}, the two types are not compliant. Indeed, the configuration above is a terminal state, and $\texttt{T1}$ and $\texttt{T2}$ are clearly different from $\texttt{end}$.}

 The scenario in Fig.~\ref{fig:scenario}(c) is rendered by
the following implementation of the service type:
\begin{lstlisting}[frame=single,basicstyle=\scriptsize\ttfamily]	
eq Tservice' = ?[str]. ![int]. cmt. ![str]. 
               brn[brnEl('hd, ![str]. ![str]. end); 
                   brnEl('sd, ![str]. ![str]. end)] .
\end{lstlisting}
In this case, as expected, the $\texttt{search}$ command returns: 
\begin{lstlisting}[frame=single,basicstyle=\scriptsize\ttfamily]
No solution.
\end{lstlisting}
meaning that types $\texttt{Tuser}$ and $\texttt{Tservice'}$ are compliant. 
Finally, the $\texttt{search}$ command applied to the type configuration 
related to the scenario depicted in Fig.~\ref{fig:scenario}(d) returns 
a solution, meaning that in that case the user and service types are not compliant. 
\end{exa}

\section{Properties of \cherry}
\label{sec:properties}

This section presents the results regarding the properties of \cherry. 
The statement of some properties exploits 
labelled transitions that permit to easily distinguish the execution of commit and rollback 
actions from the other ones. 
To this end, we can instrument the reduction semantics of collaborations by means of labels of the form
$\commitLab \,s$, $\rollLab \,s$ and $\abortLab \,s$, indicating the  rule used to derive the reduction and the session on which such operation has been done. When we do not want to distinguish the applied rule we will write 
$C\instr{s}\fwbwred C'$ to indicate that a reduction is taking plase on session $s$.


\subsection{Useful lemmata}
We now introduce a couple of results that will allow us to focus on one session per time, even if there could be multiple concurrent sessions running together.
The first of such results tells us that 
reductions taking place on different sessions can be swapped. Formally:

\begin{lem}[Swap Lemma]\label{lm:swp}
Let $C$ be a  collaboration and $s$ and $r$ two sessions. If $C\instr{s}{\fwbwred} C_1 \instr{r}{\fwbwred} C_2$ then
there exists a collaboration $C_3$ such that $C\instr{r}{\fwbwred} C_3 \instr{s}{\fwbwred} C_2$.
\end{lem}
\begin{proof}
By case analysis on the reductions $\instr{s}{\fwbwred} $ and $\instr{r}{\fwbwred}$.
\end{proof}

All the reductions taking place on a session $s$ can be put all together in sequence. Formally:

\begin{lem}\label{lm:rearr}
Let $C$ be a collaboration. 
If $C\fwbwred^{*}C_1$, then for any session $s$ in $C_1$ there exists a collaboration $C_0$
such that 
$C \fwbwred^{*}C_0\instr{s}{\fwbwred}{}^{\!\!*}\, C_1$ and $s$ is never used in the trace $C \fwbwred^{*}C_0$.
\end{lem}
\begin{proof}
By induction on the number $n$ of reduction on $s$. If there are no reductions then the thesis trivially holds.
Otherwise, we can take the very last reduction on $s$, that is the closest one to $C_1$ and iteratively apply
 Lemma~\ref{lm:swp} in order to bring it to the very end. Then we can conclude by induction on a trace with less occurrences of reductions on $s$.
\end{proof}

Thanks to Lemma~\ref{lm:rearr} we can rearrange any trace as a sequence of \textit{independent} sessions.
Moreover, given an initial collaboration $C$, for
 any reachable collaboration $C_1$ and session $s$ such that
  $C\fwbwred^{*} C_1 \instr{s}{\fwbwred}{}^{\!\!*}\, $ and $s\not \in \fwbwred^{*}$, we indicate $C_1$ 
  as the \textit{initial} collaboration for $s$. 
  This will allow us to focus just on a sigle session, say $s$, and to consider collaboration initial for $s$ without loosing of generality.

\subsection{Rollback properties}
We show some properties 
concerning the reversible behaviour of \cherry\ related to the interplay between  rollback and  commit primitives.
%
\modiff{Notably,}
Theorem~\ref{th:trace} and Lemma~\ref{lemma:safe_roll}, are an adaptation of typical properties of reversible calculi, while Lemma \ref{lemma:determinism} and Lemma \ref{lemma:undoability} are brand new.

\modiff{The following two lemmata are the key ingredients for proving Theorem~\ref{th:trace}. Specifically, the former lemma states that an abort leads back to the initial collaboration, while the latter one states that a rollback leads back to the last committed checkpoint. 
}

\begin{lem}\label{lemma:init}
Let $C$ be an initial collaboration such that
$C \fwbwred^{*} C_1$. If $C_1\abt C_2$ then $C_2 \congr C$.
\end{lem}
\begin{proof}
Since $C$ is \textit{initial}, without losing of generality we can assume
$
C\congr 		\requestAct{a}{x_1}{P_1} \mid \acceptAct{a}{x_2}{P_2} 
$.
 %
The first reduction of $C \fwbwred^{*} C_1$ has to be an application of rule \rulelabel{F-Con}, that is
\begin{align*}
C \fwred & 	\singleSession{s}{(\requestAct{a}{x_1}{P_1} \mid \acceptAct{a}{x_2}{P_2})} \\
&        (\logged{\ce{s}}{P_1\subst{\ce{s}}{x_1} }	P_1\subst{\ce{s}}{x_1}
        \ \mid\ 
	\logged{s}{P_2\subst{s}{x_2}}	P_2\subst{s}{x_2} ) 
	= C' 
\end{align*}
and, by hypothesis, $C' \fwbwred^{*} C_1$.	
Now, no matter the shape of processes in $C_1$, by applying rule \rulelabel{B-Abt}, 
we will go back to $C$, that is $C_1 \abt C$, as desired.
\end{proof}

\begin{lem}\label{lemma:fw_trace}
Let $C$ be a reachable collaboration, such that $C\cmt C_1$. If $C_1\fwred^{*}C_2 \rll C_3$ and there
is no commit in $C_1\fwred^{*}C_2$, then $C_3 \congr C_1$.
\end{lem}
\begin{proof}
Since $C$ is a reachable collaboration, this implies it has been generated from an initial 
collaboration $C_0$. Without losing of generality, similarly to the Lemma~\ref{lemma:init}'s proof, we can assume $C\congr \singleSession{s}{C_0}(\logged{\genSession}{P_1} P_2 \mid \logged{\genSession}{Q_1} Q_2)$. 
Therefore, 
\modiff{assuming w.l.o.g. that 
$P_2$ is the committing process (i.e., $P_2 \auxrel{\commitLab} P_2'$),
we have that $C_1=
 \singleSession{s}{C_0}(\logged{\genSession}{P_2'} P_2' \mid \logged{\genSession}{Q_2} Q_2) 
$.}
By hypothesis, there is no commits in $C_1 \fwred^{*} C_2$, and this implies that the log part
of $C_1$ will never change. Hence, we have that
\modiff{$C_2 \congr  \singleSession{s}{C_0}(\logged{\genSession}{P_2'} P \mid \logged{\genSession}{Q_2} Q)$}
for some processes $P$ and $Q$. By applying \rulelabel{B-Rll} 
we have that
\modiff{$C_2 \rll  \singleSession{s}{C_0}(\logged{\genSession}{P_2'} P_2' \mid \logged{\genSession}{Q_2} Q_2) \congr C_1$,}
as desired.
\end{proof}

The following theorem states that any reachable collaboration is also a 
\textit{forward only} reachable collaboration.
This means that all the states a collaboration reaches via mixed executions (also involving backward reductions)
are states that we can reach from the initial configuration with just forward reductions. 
This assures us that if the system goes back it will reach
previous visited states.

\begin{thm}\label{th:trace}
Let $C_0$ be an initial collaboration. If $C_0 \fwbwred^{*}C_1$ then $C_0 \fwred^{*}C_1$.
\end{thm}
\begin{proof}
By induction on the number $n$ of backward reductions contained into $C_0 \fwbwred^{*}C_1$.
 The base case
($n=0$)
trivially holds. In the inductive case, let us take the backward reduction which is the
nearest to $C_0$. That is:
$$C_0 \fwred^{*} C' \bwred C''\fwbwred^{*}C_1$$
Depending whether it is an $\abt$ or a $\rll$ we can apply respectively Lemma~\ref{lemma:init}
or Lemma~\ref{lemma:fw_trace} to obtain a forward trace of the form
$$C_0 \fwred^{*} C''\fwbwred^{*}C_1$$
and we can conclude by applying the inductive hypothesis on  the obtained trace which contains
less backward moves with respect to the original one.
\end{proof}
\medskip

We now show a variant of the so-called Loop Lemma~\cite{DanosK04}. In a fully reversible
calculus this lemma states that each computational step, either forward or backward, 
can be undone. Since reversibility in \cherry\ is controlled, we have to state that if a 
reversible step is possible (e.g., a rollback is \textit{enabled}) then the effects
of the rollback can be undone. 

\begin{lem}[Safe rollback]\label{lemma:safe_roll}
Let $C_1$ and $C_2$ be reachable collaborations. 
If \mbox{$C_1 \bwred C_2$} then \mbox{$C_2 \fwred^{*} C_1$}.
\end{lem}

\begin{proof}
Since $C_1$ is a reachable collaboration, we have that there exists an initial collaboration $C_0$
such that $C_0\fwbwred^{*}C_1$.
By applying Theorem~\ref{th:trace} we can rearrange the trace 
such that it contains just forward transitions as follows
$$C_0\fwred^{*} C_1$$
If the backward reduction is obtained by applying \rulelabel{B-Abt},
by Lemma~\ref{lemma:init} we have $C_2\congr C_0$, from which the thesis trivially follows.
Instead, if the backward reduction is obtained by applying \rulelabel{B-Rll},
we proceed by case analysis depending on the presence of commit reductions in the trace.
If they are present, we select the last of such commit, that is we can decompose the trace in the following way:
$$C_0 \fwred^{*}\ \cmt C_{cmt} \fwred^{*}C_1\rll C_2$$
and by applying Lemma~\ref{lemma:fw_trace} we have that $C_2 \fwred^{*}\congr C_1$ as desired.

In the case there is no commit in the trace, 
without losing of generality we can assume
\modiff{$C_0\congr\requestAct{a}{x_1}{P_1} \mid \acceptAct{a}{x_2}{P_2}$ and we have:
$$
C_0
 \fwred
 \singleSession{s}{C_0}(\logged{\genSession}{P_1'} P_1' \mid \logged{\genSession}{P_2'} P_2')
 \fwred^{*}
 \singleSession{s}{C_0}(\logged{\genSession}{P_1'} P_1'' \mid \logged{\genSession}{P_2''} P_2')\congr C_1
$$}
By rule \rulelabel{B-Rll}, we have 
\modiff{$C_2 \congr \singleSession{s}{C_0}(\logged{\genSession}{P_1'} P_1' \mid \logged{\genSession}{P_2'} P_2').$}
Thus, we can conclude by noticing that $C_0\fwred C_2$ 
\modiff{is} the first reduction in $C_0\fwred^{*} C_1$.
\end{proof}

A rollback always brings the system 
to the last taken checkpoint.
We recall that, since there may be sessions running in parallel, a collaboration may be able to do different rollbacks within different sessions. Thus, determinism only holds relative to a given session, and rollback within one session has no effect on any other parallel session.

\begin{lem}[Determinism]\label{lemma:determinism}
Let C be a reachable collaboration. 
\added{If \mbox{$C \rlls{} C'$} and \mbox{$C \rlls{} C''$} then \mbox{$C' \congr C''$}.}
\end{lem}

\begin{proof}
Since $C$ is a reachable collaboration, it is has been generated by an initial collaboration
$C_0$ of the form $C_0 = \requestAct{a}{x}{P_1}   \mid \acceptAct{a}{x}{P_2}$, 
and by Theorem~\ref{th:trace} we have that
$C_0 \fwred^{*} C$. We distinguish two cases, whether in the trace
there has been at least one commit or not. In the first case, we can decompose the trace in such a way
to single out the last commit as follows:
$$C_0 \fwred^{*}C_{cmt} \fwred^{*} C$$
so that in the reduction $C_{cmt} \fwred^{*} C$ there is no commit. 
If from $C$ the rollbacks $C\rlls{} C'$ and $C\rlls{} C''$ are triggered by the same process, 
the thesis trivially follows.
In the other case, we have that:
$$
C \congr \singleSession{s}{C_0}(\logged{\genSession}{Q_1} Q_1' \mid \logged{\genSession}{Q_2} Q_2') 
$$
with both $Q_1'$ and $Q_2'$ able to trigger a rollback. If the roll action is executed by $Q_1'$, by applying rule \rulelabel{B-Rll} we have that
$$
\singleSession{s}{C_0}(\logged{\genSession}{Q_1} Q_1' \mid \logged{\genSession}{Q_2} Q_2') 
 \rll 
\singleSession{s}{C_0}(\logged{\genSession}{Q_1} Q_1 \mid \logged{\genSession}{Q_2} Q_2) 
\congr C'$$
If the roll is triggered by $Q_2'$, , by applying rule \rulelabel{B-Rll} up to structural congruence we have that
$$
\singleSession{s}{C_0}(\logged{\genSession}{Q_1} Q_1' \mid \logged{\genSession}{Q_2} Q_2') 
 \rll
\singleSession{s}{C_0}(\logged{\genSession}{Q_1} Q_1 \mid \logged{\genSession}{Q_2} Q_2) 
\congr C''$$
We can conclude by noticing that $C' \equiv C''$, as desired.
\end{proof}

The last rollback property states that a  collaboration cannot go back 
 to a state prior to the execution of a commit action, that is commits have a 
 persistent effect. Let us note that recursion does not affect this theorem, since at the beginning of a collaboration computation 
 there is always a new session establishment, leading to a stack of past configurations. Hence it is never the case that from a collaboration $C$ you can reach again $C$ via forward steps.
\begin{thm}
[Commit persistency]\label{lemma:undoability}
Let C be a reachable collaboration. If $C \cmts{} C'$ then there exists no $C''$ such that
$C' \fwred^{*}\,\rlls{} C''$ and $C'' \fwred^{+} C$.
\end{thm}

\begin{proof}
We proceed by contradiction. Suppose that there exists $C''$ such that 
$C' \fwred^{*}\rlls{} C''$ and $C'' \fwred^{+} C$.
Since $C$ is a reachable collaboration, thanks to Theorem~\ref{th:trace}
we have that there exists an initial collaboration $C_0$ 
such that $C_0\fwred^{*}C\cmts{s} C'$. Since a rollback brings back the 
collaboration to a point before a commit, this means it has been restored 
\modiff{a checkpoint committed before 
the last one}
(by \rulelabel{B-Rll}, indeed, only processes stored in logs can be 
\modiff{restored}). 
This implies that there exist at least two different commits in the trace
such that
$$C_0 \fwred^{*}\modiff{\cmts{}}
C_{cmt}\fwred^{*}C \cmts{} C'\fwred^{*}\modiff{C_{rl}}\rlls{} C''$$
with $C'' = C_{cmt}$. 
Now, 
consider the commit performed by $C$,
supposing that it is triggered by
$P_c$ evolving to $P'_c$ in doing that,  we have that:
\begin{align*}
&C \congr 	\singleSession{s}{C_0}(\logged{\genSession}{P} P_c \mid \logged{\genSession}{Q} Q_c) 
\;\text{ and }\; C'= \singleSession{s}{C_0}(\logged{\genSession}{P'_c} P'_c \mid \logged{\genSession}{Q_c} Q_c)
\end{align*}
Now, 
let us consider the case where
$C'\fwred^{*} C_{rl}$ without any commit being present in the trace, 
hence:
$$C_{rl} \congr 	\singleSession{s}{C_0}(\logged{\genSession}{P'_c} P_{rl} \mid \logged{\genSession}{Q_c} Q_{rl}) $$
By hypothesis, from $C_{rl}$ a rollback is possible. 
Regardless the rollback is triggered by $P_{rl}$ or  $Q_{rl}$, we have that
$C_{rl} \rll C'$, \modiff{hence $C'\congr C''$.}
Now, from $C'$ we cannot reach $C$ (i.e., $C'\fwred^{+}\!\!\!\!\!\!\!\!\!\!/\ \ \ \  C$), 
as $C'$ is derived from $C$ and the rollback can only bring the collaboration back to $C'$.
Therefore, we have $C''\fwred^{+}\!\!\!\!\!\!\!\!\!\!/\ \ \ \  C$,
which violates the initial hypothesis.

Let us consider, instead, the case where 
$C'\fwred^{*} C_{rl}$ with at least a commit being present in the trace. By applying rule \rulelabel{B-Rll}, from $C_{rl}$ the rollback will lead to the last committed collaboration $C_{cmt}'$ such that 
$C' \fwred^{*}\cmts{s}
C_{cmt}'\fwred^{*}C_{rl}$.
Hence, $C_{cmt}'\congr C''$ and, like in the previous case,
from $C_{cmt}'$ we cannot reach $C$. 
Again we have $C''\fwred^{+}\!\!\!\!\!\!\!\!\!\!/\ \ \ \  C$,
which violates the initial hypothesis.
Therefore, in each case the initial hypothesis is violated, and hence we conclude. 
\end{proof}

\subsection{Soundness properties}
The second group of properties concerns soundness guarantees. 
The definition of these properties requires formally characterising the errors that 
may occur in the execution of an unsound collaboration. We rely on error reduction 
(as in \cite{ChenDSY17}) rather than on the usual static characterisation of errors
(as, e.g., in \cite{YoshidaV07}), since rollback errors cannot be easily detected statically. 
In particular, we extend the syntax of \cherry\ collaborations
with the $\rollError$ and $\comError$ terms, denoting respectively 
collaborations in rollback and communication error states:
$$
\begin{array}{@{}r@{ \ }c@{\ }l@{}}
\\[-.8cm]
C & ::= &
\ldots
\, \mid\,
\graybox{$\logged{\genSession}{\checkpointType{P_1}}P_2$}    
\mid
\graybox{$\rollError$}
\mid
\graybox{$\comError$}
\\[-.3cm]
\end{array}
$$
where $\checkpointType{P}$ denotes either a checkpoint $P$ committed by the party 
or a checkpoint $\imposed{P}$ imposed by the other party of the session. 
%
%
%
%
\begin{figure*}[t]
	\centering
	\small
	\modif{
	\begin{tabular}{@{}c@{}}
	$
	\infer[$\rulelabel{E-Com$_1$}$]{
	\logged{{\genSession}}{\checkpointType{Q_1}} P_1 
	\mid 
	\logged{\genSession}{\checkpointType{Q_2}} P_2
	\fwred
	\comError
	}
	{
	P_1 \auxrel{\send{{k}}{v}} P_1'
	\quad
	\neg P_2 \BarbInp{k}
        \ \
	\added{\neg P_2 \Barb{roll}}
	}
	$		
	\quad\ 
	$
	\infer[$\rulelabel{E-Com$_2$}$]{
	\logged{{\genSession}}{\checkpointType{Q_1}} P_1 
	\mid 
	\logged{\genSession}{\checkpointType{Q_2}} P_2
	\fwred 
	\comError
	}
	{
	P_1 \auxrel{\receive{{k}}{x}} P_1'
	\quad
	\neg P_2 \BarbOut{k}
        \ \ 
        \added{\neg P_2 \Barb{roll}}
	}
	$	
	\\[.3cm]
	$
	\infer[$\rulelabel{E-Lab$_1$}$]{
	\logged{\ce{\genSession}}{\checkpointType{Q_1}} P_1 
	\mid 
	\logged{\genSession}{\checkpointType{Q_2}} P_2
	\fwred
	\comError
	}
	{
	P_1 \auxrel{\select{{k}}{l}} P_1'
	\quad
	\neg P_2\BarbBranching{k}{l}
        \ \ \added{\neg P_2 \Barb{roll}}
	}
	$			
	\qquad
	$
	\infer[$\rulelabel{E-Lab$_2$}$]{
	\logged{\ce{\genSession}}{\checkpointType{Q_1}} P_1 
	\mid 
	\logged{\genSession}{\checkpointType{Q_2}} P_2
	\fwred
	\comError
	}
	{
	P_1 \auxrel{\branching{k}l} P_1'
	\quad
	\neg P_2\BarbSelect{k}{l}
        \ \ \added{\neg P_2 \Barb{roll}}
	}
	$
		\\[.3cm]	
	$
	\infer[$\rulelabel{E-Cmt$_1$}$]{
	\logged{\ce{\genSession}}{\checkpointType{Q_1}} P_1 
	\mid
	\logged{\genSession}{\checkpointType{Q_2}} P_2
	\fwred 
	\logged{\ce{\genSession}}{P_1'} P_1' 
	\mid
	\logged{\genSession}{\imposed{P_2}} P_2
	}
	{
	P_1 \auxrel{\commitLab} P_1'
	& & 
	\checkpointType{Q_2} \neq P_2	
	}
	$	
	\\[.3cm]
	$
	\infer[$\rulelabel{E-Cmt$_2$}$]{
	\logged{\ce{\genSession}}{\checkpointType{Q_1}} P_1 
	\mid
	\logged{\genSession}{\checkpointType{Q_2}} P_2
	\fwred 
	\logged{\ce{\genSession}}{P_1'} P_1' 
	\mid
	\logged{\genSession}{\checkpointType{Q_2}} P_2
	}
	{
	P_1 \auxrel{\commitLab} P_1'
	& & 
	\checkpointType{Q_2} = P_2
	}
	$	
	\\[.3cm]		
	$
	\infer[$\rulelabel{E-Rll$_1$}$]{
	\logged{\ce{\genSession}}{Q_1} P_1 
	\mid
	\logged{\genSession}{\checkpointType{Q_2}} P_2
	\bwred
	\logged{\ce{\genSession}}{Q_1} Q_1
	\mid
	\logged{\genSession}{\checkpointType{Q_2}} Q_2
	}
	{	
	P_1 \auxrel{\rollLab} P_1'		
	}
	$	
	\\[.3cm]	
	$
	\infer[$\rulelabel{E-Rll$_2$}$]{
	\logged{\ce{\genSession}}{\imposed{Q_1}} P_1 
	\mid
	\logged{\genSession}{\checkpointType{Q_2}} P_2
	\fwred 
	\rollError
	}
	{
	P_1 \auxrel{\rollLab} P_1'	
	}
	$
	\\[.1cm]
	\hline
	\end{tabular}}	
	\caption{\Cherry\ semantics: error reductions.}
	\label{tab:error_red}
\end{figure*}

The semantics of \cherry\ is extended as well 
\modif{by the additional
error reduction rules in  Fig.~\ref{tab:error_red},
where \rulelabel{E-Cmt$_1$} and \rulelabel{E-Cmt$_2$} replace \rulelabel{F-Cmt},
and \rulelabel{E-Rll$_1$} replaces  \rulelabel{B-Rll}; moreover,
$\checkpointType{\ }$ is used in the checkpoints of the other rules.
}
The error semantics does not affect the normal behaviour of \cherry\ specifications, 
but it is crucial for stating our soundness theorems. 
Its definition is based on the notion of \emph{barb} predicate:
$P\Barb{\mu}$ holds if there exists $P'$ such that $P \Rightarrow P'$ and 
$P'$ can perform an action $\mu$, 
where $\mu$ stands for $k?$, $k!$ ,$k\triangleleft l$, $k\triangleright l$, or $roll$
(i.e., input, output, select, branching action along session channel $k$, \added{or roll action});
$\Rightarrow$ is the reflexive and transitive closure of $\auxrel{\ite}$.
The meaning of the error semantics rules is as follows.
A \emph{communication error} takes place in a collaboration when 
a session participant is willing to perform an output but the other participant is 
\added{ready to perform neither the corresponding input nor a roll back}
(rule \rulelabel{E-Com$_1$}) or vice versa
\modif{(rule \rulelabel{E-Com$_2$}),} 
or 
one participant is willing to perform a selection but the corresponding branching is not available 
on the other side \modif{(rule \rulelabel{E-Lab$_1$})}
or viceversa \modif{(rule \rulelabel{E-Lab$_2$}).}
Instead, a \emph{rollback error} takes place in a collaboration when a participant is willing 
to perform a rollback action but her checkpoint has been imposed by the other participant
(rule \rulelabel{E-Rll$_2$}). To enable this error check, the rules for commit and rollback 
\modif{(rules \rulelabel{E-Cmt$_1$}, \rulelabel{E-Cmt$_2$}, and  \rulelabel{E-Rll$_1$})}
have been modified to keep track of imposed overwriting of checkpoints. 
This information is not relevant for the runtime execution of processes, 
but it is necessary for characterising 
the rollback errors that our 
type-based approach prevents. 

Besides defining the error semantics, we also need to define erroneous collaborations,
based on the following notion of context:
$
\collContext \ ::=  \ [\cdot] \,\, \mid\,\,  \collContext \!\mid\! C \,\, \mid\,\,  \singleSession{s}{C}\ \collContext
$.

\begin{defi}[Erroneous collaborations]
A collaboration $C$ is \emph{communication (resp. rollback) erroneous} if 
$C = \collContext[\comError]$ (resp. $C = \collContext[\rollError]$).
\end{defi}
The key soundness results follow: a rollback safe collaboration never 
reduces to either a rollback erroneous collaboration (Theorem~\ref{th:roll_soundness}) or a
communication erroneous collaboration (Theorem~\ref{th:com_soundness}). 

\begin{thm}[Rollback soundness]
\label{th:roll_soundness}
If $C$ is a \rollSafe\ collaboration, then we have that $C\ {\fwbwred\!\!\!\!\!\!/\ }^*$\ $\collContext[\rollError]$.
\end{thm}

\modiff{
\begin{proofS} 
The proof proceeds by contradiction 
(the full proof is reported in Appendix~\ref{proofs:binary}).
\end{proofS}}

\bigskip

\begin{thm}[Communication soundness]
\label{th:com_soundness}
If $C$ is a \rollSafe\ collaboration, then we have that $C\ {\fwbwred\!\!\!\!\!\!/\ }^*\ \collContext[\comError]$.
\end{thm}

\modiff{
\begin{proofS} 
The proof proceeds by contradiction 
(the full proof is reported in Appendix~\ref{proofs:binary}).
\end{proofS}}

\bigskip

We conclude with a  progress property of \cherry\ sessions: 
given a rollback safe collaboration that can initiate a session, 
each collaboration reachable from it either is able to progress on the session 
with a forward/backward reduction step or has correctly reached the end of 
the session. This result follows from 
Theorems~\ref{th:roll_soundness} and~\ref{th:com_soundness}, and from
the fact that we consider binary sessions without delegation and subordinate sessions. 

\begin{thm}[Session progress]
\label{th:deadlock_freedom}
Let $C=(\requestAct{a}{x_1}{P_1} \mid \acceptAct{a}{x_2}{P_2})$ be a \rollSafe\ collaboration.
If $C \!\fwbwred^*\! C'$ then either $C' \!\fwbwred\! C''$ for some $C''$ or 
$C' \!\congr\! \singleSession{s\!}{\!C}
(\logged{{\genSession}}{\checkpointType{Q_1\!}}\! \inact 
\!\mid \!
\logged{\genSession}{\checkpointType{Q_2\!}}\! \inact)$
for some $\checkpointType{Q_1}$ and $\checkpointType{Q_2}$.
\end{thm}
\modiff{\begin{proof} 
The proof proceeds by contradiction.
Suppose that $C$ is rollback safe and 
$C \fwbwred^* C'$ with $C' \fwbwred\!\!\!\!\!\!/\ \ $ and 
$C' \congr\!\!\!\!\!/\ \  
\singleSession{s}{C}
(\logged{{\genSession}}{\checkpointType{Q_1}} \inact 
\ \mid \ 
\logged{\genSession}{\checkpointType{Q_2}} \inact)$
for any $\checkpointType{Q_1}$ and $\checkpointType{Q_2}$.
The only situations that prevents $C'$ from progressing 
are $C'= \collContext[\rollError]$ or $C'= \collContext[\comError]$. 
However, from Theorems~\ref{th:roll_soundness} and~\ref{th:com_soundness}, respectively, 
we have $C'\neq \collContext[\rollError]$ and $C'\neq \collContext[\comError]$, 
which is a contradiction.  
\end{proof}}

%


\section{\Cherry\ at work on a speculative parallelism scenario}
\label{sec:case_study}

To shed light on the practical effectiveness of \cherry\ and the related notion 
of rollback safety, we consider in this section a simple, yet realistic, scenario concerning a form 
of speculative execution borrowed from \cite{speculative_exec}.
%
In this scenario, \emph{value speculation} is used as a mechanism for 
increasing parallelism, hence system performance, by predicting 
values of data dependencies between tasks. 
Whenever a value prediction is incorrect, corrective actions must be taken 
in order to re-execute the data consumer code with the correct data value.
In this regard, as shown in \cite{GiachinoLMT17} for a shared-memory setting, 
reversible execution can permit to relieve programmers from the burden of 
properly undoing the actions subsequent to an incorrect prediction.
Here, we tailor the scenario to the channel-based communication model of 
session-based programming, and show how our rollback safety checking supports
programmers in identifying erroneous rollback recovery settings. 

\begin{figure}[t]
\centering
\includegraphics[scale=.39]{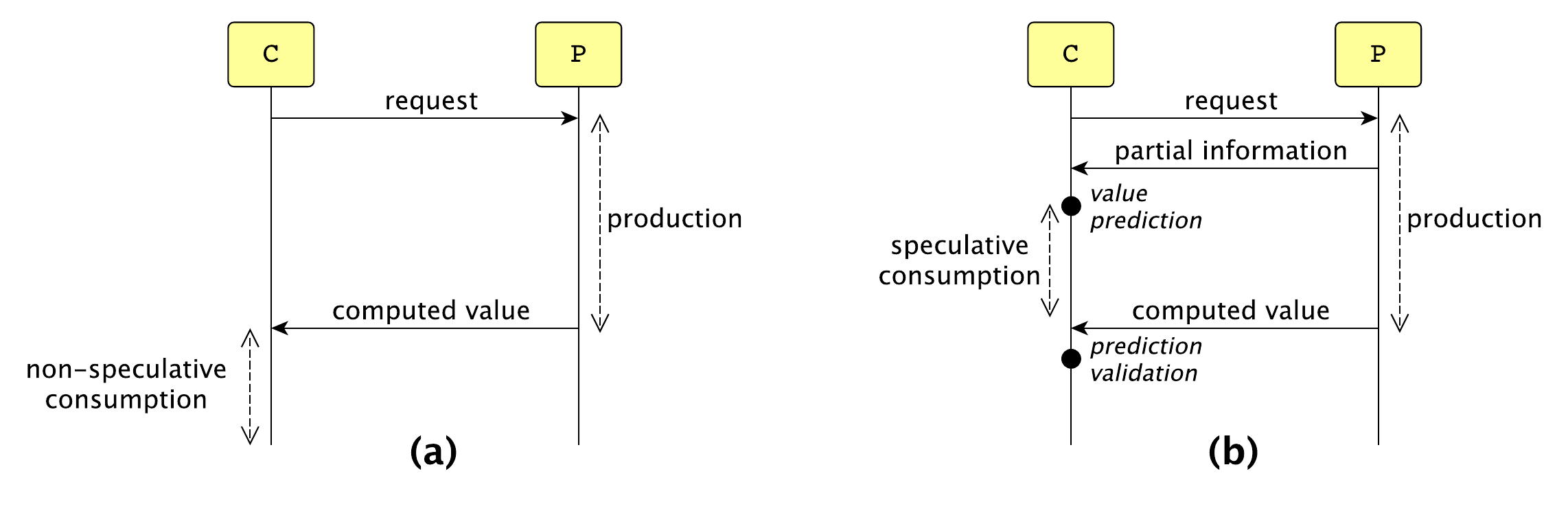}
\vspace*{-.4cm}
\caption{Producer-consumer scenario with non-speculative (a) 
and speculative (b) consumers.}
\label{fig:speculative}
\end{figure}

In the producer-consumer scenario depicted in Fig.~\ref{fig:speculative}(a) the session 
participant \texttt{P} produces a value and the participant \texttt{C} consumes it. The data 
dependence between \texttt{P} and \texttt{C} serialises their executions, thus forcing \texttt{C} 
to wait for the completion of the value production that requires a fairly long time. In the 
scenario in Fig.~\ref{fig:speculative}(b), instead, \texttt{C} enacts a speculative behaviour, as it 
predicts ahead of time the value computed by \texttt{P} from a partial information. By using the 
predicted value, \texttt{C} can execute speculatively and concurrently with \texttt{P}. When 
\texttt{P} completes the production, \texttt{C} validates the prediction by comparing the actual 
value computed by \texttt{P} and the predicted one; if the prediction is precise, we gain 
performance because the execution of \texttt{C} and \texttt{P} overlapped in time, otherwise 
rollback is used to move \texttt{C} and \texttt{P} back to a state that precedes the speculative 
behaviour, in order to re-execute \texttt{C} using the correct value. The behaviours of \texttt{C} 
and \texttt{P} can be recursively defined in order to repeat the overall execution once a value 
is correctly consumed. 

The scenario informally described above is rendered in \cherry\ as 
$$
\requestPrefix{start}{x}. P_{\texttt{C}} \mid \acceptPrefix{start}{y}. P_{\texttt{P}}
$$
where the consumer and producer processes are:
$$
\begin{array}{r@{\ }c@{ \ }l}
P_{\texttt{C}} & = & 
\recAct{X}{\,}
\sendAct{x}{f_{\textit{\textsf{req}}}()}{\,}
\branchAct{x}{\
\branch{l_{\textit{\textsf{spec}}}}{
\receiveAct{x}{x_{\textit{\textsf{partial}}}:\strType}{\,}
\receiveAct{x}{x_{\textit{\textsf{final}}}:\strType}{\,}\\
&&
\hspace*{4.6cm}
\ifPi\ (f_{\textit{\textsf{compare}}}(x_{\textit{\textsf{partial}}},x_{\textit{\textsf{final}}}))\
\thenPi\ \roll\ 
\elsePi\ \commitAct{\,}X
}
\, \branchSep \\
&&
\hspace*{3.4cm}
\branch{l_{\textit{\textsf{nonSpec}}}}{
\receiveAct{x}{x_{\textit{\textsf{computed}}}:\strType}{\,}
\commitAct{\,}
X}\
}
\\[.3cm]
P_{\texttt{P}} & = & 
\recAct{Y}{\,}
\receiveAct{y}{y_{\textit{\textsf{req}}}:\strType}{\,}\\
&&
\hspace*{.7cm}
\ifPi\ (f_{\textit{\textsf{eval}}}(y_{\textit{\textsf{req}}}))\
\thenPi\
\selectAct{{y}}{l_{\textit{\textsf{spec}}}}{\,}
\sendAct{y}{f_{\textit{\textsf{partial}}}(y_{\textit{\textsf{req}}})}{\,}
\sendAct{y}{f_{\textit{\textsf{final}}}(y_{\textit{\textsf{req}}})}{\,}
Y
\\
&&
\hspace*{.7cm}
\elsePi\
\selectAct{{y}}{l_{\textit{\textsf{nonSpec}}}}{\,}
\sendAct{y}{f_{\textit{\textsf{compute}}}(y_{\textit{\textsf{req}}})}{\,}
Y
\end{array}
$$
The producer evaluates each consumer's request in order to establish whether to provide 
directly the produced value or the partial information for the prediction.
In the former case the consumer commits the session and both participants restart, 
while in the latter one the consumer commits or rolls back depending on the result of the 
comparison between the predicted value and the produced one. 


\modif{To check whether the above collaboration is rollback safe, we have to check compliance 
between the session types of the two parties. 
The session types \texttt{Tc} and \texttt{Tp} associated by the type system to the consumer and the producer 
processes, respectively, and the corresponding initial type configuration 
are as follows (hereafter we use the \maude{} implementation of the \cherry\ 
type syntax):}
\begin{lstlisting}[frame=single,basicstyle=\scriptsize\ttfamily]
eq Tc = mu X . ![str]. brn[brnEl('spec, ?[str]. ?[str]. (roll (+) (cmt. X))); 
                           brnEl('nonSpec, ?[str]. cmt. X)] .

eq Tp = mu Y . ?[str]. ((sel['spec]. ![str]. ![str]. Y) (+) (sel['nonSpec]. 
                                                             ![str]. Y)) .

eq CPInitConfig = init(Tc,Tp) ckp(Tc) > Tc || ckp(Tp) > Tp .
\end{lstlisting}
\modif{As discussed in Sec.~\ref{sec:maude}, the compliance check for the above type specification 
is performed by resorting to the \maude\ command $\texttt{search}$ as follows:}
\begin{lstlisting}[frame=single,basicstyle=\scriptsize\ttfamily]	
search 
  CPInitConfig  
  =>! 
  init(T:Type,T':Type) CT1:CkpType > T1:Type || CT2:CkpType > T2:Type
such that T1 =/= end or T2 =/= end .
\end{lstlisting}	
The above command returns: 
\begin{lstlisting}[frame=single,basicstyle=\scriptsize\ttfamily]
No solution.
\end{lstlisting}
\modif{meaning that the producer-consumer collaboration is rollback safe.
To investigate more in detail the behaviour of this collaboration, we can 
generate the transition system associated to the type configuration
\texttt{CPInitConfig} by using the following \maude\ commands:}
\begin{lstlisting}[frame=single,basicstyle=\scriptsize\ttfamily]	
search CPInitConfig =>* C:Configuration .
show search graph .
\end{lstlisting}	
\modif{The first command searches for all states of the transition, reachable 
in none, one, or more steps, while the second command returns 
the current search graph generated by the previous search command.
A graphical representation of the generated graph is reported in Fig.~\ref{fig:graph1},
where states represent the reachable type configurations (state $0$ corresponds to \texttt{CPInitConfig})
and transitions are labelled by the applied rules (from Figure~\ref{fig:typeSemantics_ext}).  
}
\begin{figure}[t]
\centering
\includegraphics[scale=.55]{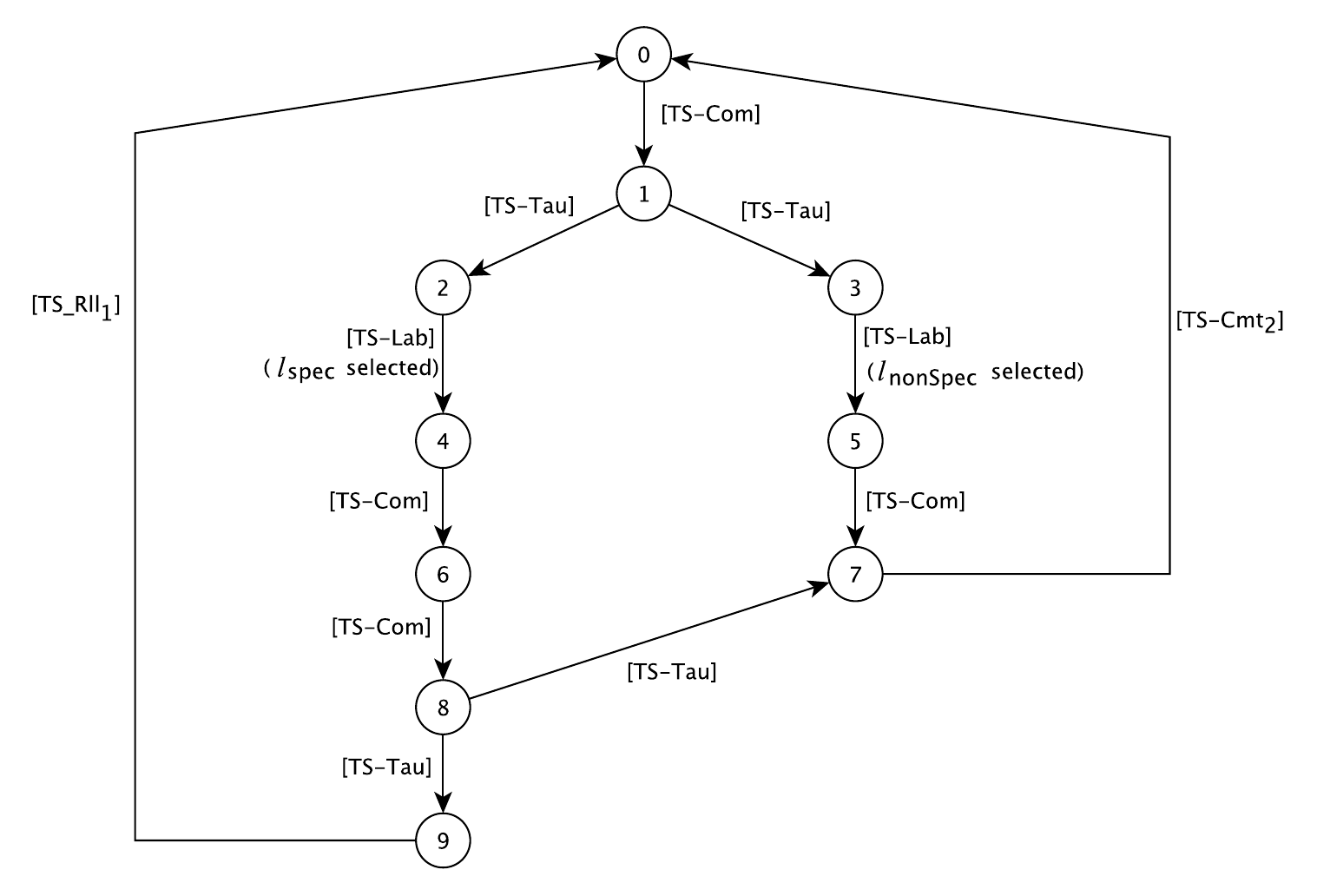}
\vspace*{-.4cm}
\caption{Transition system of the type configuration \texttt{CPInitConfig}.}
\label{fig:graph1}
\end{figure}
\modif{Notably, in case a value is correctly consumed, since the commit action is performed 
at the end of the recursive step, the type configuration resulting from the commit coincides 
with the initial one (see transition from state 7 to state 0 in the graph), as both consumer and producer
are ready to repeat their interactions for a new value. 
In case of incorrect prediction (transition from state 8 to state 9), 
instead, the session execution is moved back to the last checkpoint (transition from state 9 to state 0),
corresponding to the successfully consumption of the previous requested value that,
at type level, correspond again to the initial type configuration. Indeed, as previously discussed,
no commit occurs during the speculative consumption, hence no checkpoint has been created after 
the one recorded at the end of the previous recursive step. 
%
It is worth noticing that in case the prediction of the first produced value is wrong, and hence no commit action 
is performed by the consumer yet, according to the \cherry\ semantics the checkpoint 
corresponds to the beginning of the session. }


\modif{Let us consider now a variant of the above scenario where the producer commits each 
time a value production is completed, which could apparently seem a reasonable behaviour from 
the producer side.
The session type \texttt{Tp'} of such a producer process and the corresponding initial type configuration 
are as follows:}
\begin{lstlisting}[frame=single,basicstyle=\scriptsize\ttfamily]
eq Tp' = mu Y . ?[str]. ((sel['spec]. ![str]. ![str]. cmt. Y) 
                         (+) (sel['nonSpec]. ![str]. cmt. Y)) .

eq CPInitConfig' = init(Tc,Tp') ckp(Tc) > Tc || ckp(Tp') > Tp' .
\end{lstlisting}
\pagebreak
\modif{This time the collaboration is not rollback safe. Indeed, the compliance check
returns two solutions:}
\begin{lstlisting}[frame=single,basicstyle=\scriptsize\ttfamily]
Solution 1 (state 16)
CT1 --> ickp(roll)
T1 --> err
CT2 --> ckp(mu Y . ?[str]. ((sel['spec]. ![str]. ![str]. cmt. Y)
                            (+)sel['nonSpec]. ![str]. cmt. Y))
T2 --> err

Solution 2 (state 17)
CT1 --> ickp(roll(+)cmt. mu X . ![str]. brn[brnEl('nonSpec, ?[str]. cmt. X); 
                                brnEl('spec, ?[str]. ?[str]. (roll(+)cmt. X))])
T1 --> err
CT2 --> ckp(mu Y . ?[str]. ((sel['spec]. ![str]. ![str]. cmt. Y)
                            (+)sel['nonSpec]. ![str]. cmt. Y))
T2 --> err
\end{lstlisting}
\modif{corresponding to two erroneous configurations. The overall behaviour is 
graphically depicted by the transition system in Fig.~\ref{fig:graph2},
where state $0$ corresponds to \texttt{CPInitConfig'} and states 16 and 17
are the two erroneous configurations above.  }
\begin{figure}[t]
\centering
\includegraphics[scale=.55]{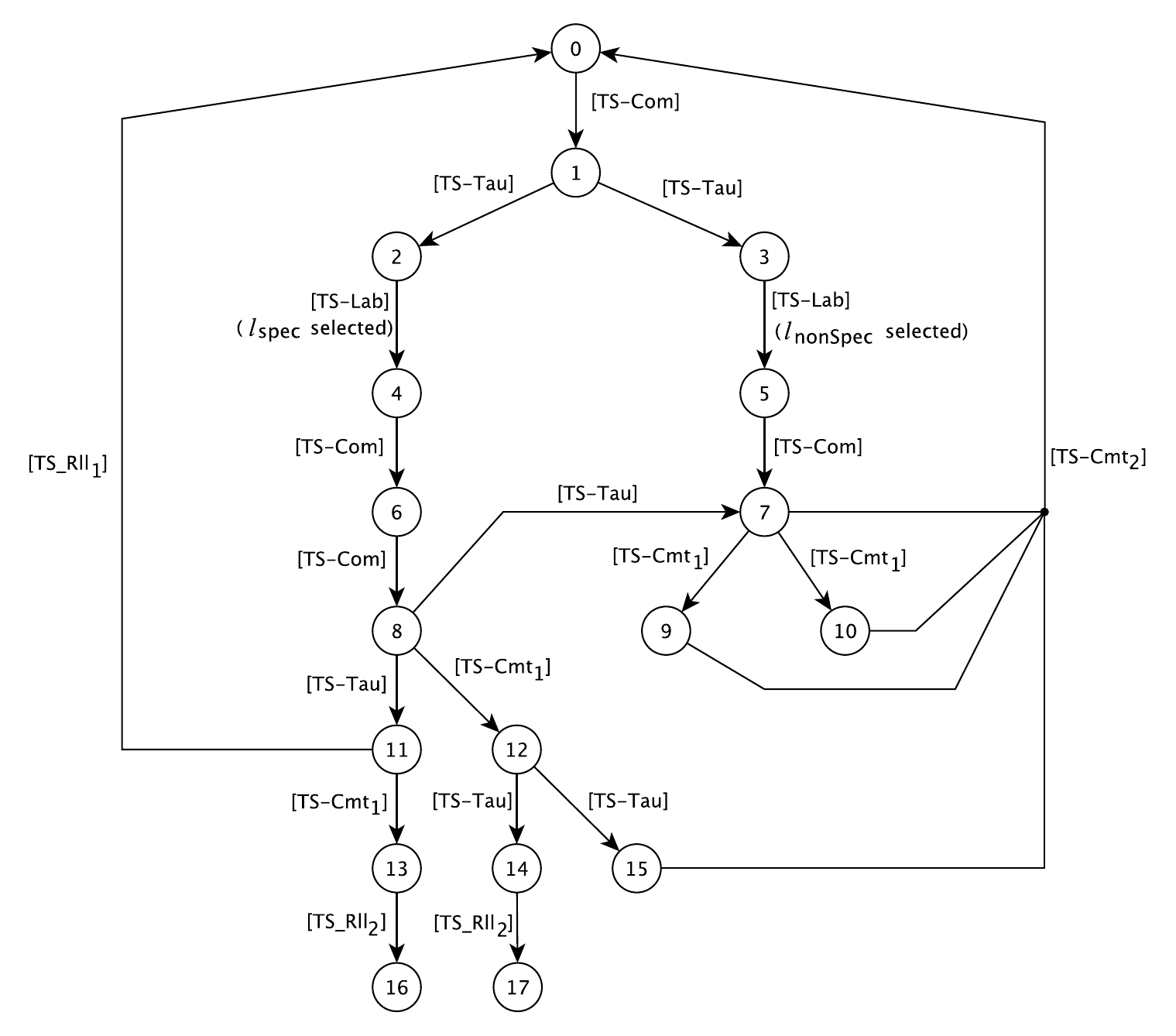}
\vspace*{-.4cm}
\caption{Transition system of the type configuration \texttt{CPInitConfig'}.}
\label{fig:graph2}
\end{figure}
\modif{While the commit action in the non-speculative case
(transitions from state 7 to state 10, and from state 9 to state 0),
 does not affect the compliance between the two session participants, 
the other commit action (transitions from state 8 to state 12, and from state 11 to state 13),
overwrites the checkpoint set by the consumer, making it impossible to 
re-execute the consumer with the correct value. This situation, undesirable for the consumer, 
is detected by our compliance check. 
}

\section{Related work}
\label{sec:rw}

In the literature we can distinguish three ways of dealing with rollback: 
either using explicit rollbacks and implicit commits~\cite{LaneseMSS11}, or by using explicit commits and spontaneous aborts~\cite{DanosK05,VriesKH10}, \modifRev{or by just
considering implicit rollbacks and commits \cite{FieldV05,SwalensKM17}.}
Differently from these works, we have introduced a way to 
control reversibility by both \textit{triggering} it and \textit{limiting} its scope. 
Reversibility is triggered via an explicit rollback primitive (as in~\cite{LaneseMSS11}), 
 while explicit commits limit the scope of potential future reverse executions (as in~\cite{DanosK05,VriesKH10}).
 Differently from~\cite{DanosK05,VriesKH10}, commit does not require any synchronisation, as it is a local decision.
 This could lead to run-time misbehaviours where a process willing to roll back to its last checkpoint reaches a point
 which has been imposed by another participant of the session.  
 Our type discipline rules out such cases.

Reversibility in behavioural types has been studied in different formalisms:
\textit{contracts}~\modifRev{\cite{BarbaneraDd14,BarbaneraDLd16,BarbaneraLd17}}, 
\textit{binary session types}~\cite{MezzinaP17a}, \textit{multiparty 
session types}~\cite{MezzinaP21,CastellaniDG17,JLAMP_TY15,TiezziY16},
and \textit{global graphs}~\cite{MezzinaT17,fmt18}.
In~\modifRev{\cite{BarbaneraDd14,BarbaneraDLd16,BarbaneraLd17}}  choices can be seen as \textit{implicit}
checkpoints and the system can go back to a previous choice and try another branch. In~\cite{BarbaneraDd14}  rollback is triggered non-deterministically (e.g., it may happen at any time during the execution), while in~\modifRev{\cite{BarbaneraDLd16,BarbaneraLd17}} it is triggered by the system only when the \modifRev{forward} computation is stuck (e.g., the client cannot further continue). 
In both works reversibility (and rollbacks) is used to achieve a relaxed variant of client-server compliance: if there exists an execution in which the client is able to terminate then  the client and server are compliant. Hence, reversibility is used as a means to explore different branches if the current one leads to a deadlock.
In~\cite{MezzinaP17a} reversibility is studied in the context of binary session types. 
Types information is used at run-time by monitors, for  binary~\cite{MezzinaP17a} and multiparty~\cite{MezzinaP21} settings, to keep track of the computational history of the system.
allowing to revert any computational step. 
 where global types are enriched with computational history. 
There, reversibility is uncontrolled, and each computational step can be undone. 
In~\cite{CastellaniDG17} global types are enriched with history information, and choices are seen as 
labelled checkpoints. 
The information about checkpoints is projected into local types. 
At any moment, the party who decided which branch to take in a choice may decide to revert it, forcing the entire system to go back to a point prior to the choice. Hence, rollback is confined inside choices and it is spontaneous 
meaning that the former can be programmed while the latter cannot. 
Checkpoints are not seen as commits, and  a rollback can bring the system to a state prior to several checkpoints. 
In~\cite{JLAMP_TY15} an \textit{uncontrolled} reversible variant of session $\pi$-calculus is presented,
while~\cite{TiezziY16} studies different notions of reversibility for both binary and multiparty single sessions. 
 In~\cite{MezzinaT17,fmt18,FrancalanzaMT20}  global graphs
are extended with conditions on branches. These conditions at runtime can trigger coordinated rollbacks to revert a distributed choice. 
Reversibility is confined into branches of a distributed choice and 
not all the computational steps are reversible; inputs, in fact,  are  irreversible unless they are
inside an ongoing loop. 
%
To trigger a rollback several conditions  and constraints about loops have to be satisfied.
Hence, in order to trigger a rollback a runtime condition should be satisfied.

The closest work in the literature to ours is
\cite{Vidal23} where a checkpoint rollback facility is studied on top a reversible actor based language. Here, checkpoint and rollbacks primitives are explicit, as in our approach. Also, the proposed approach scales with dynamically created processes (e.g., spawning), while we do not deal with this issue. On the other hand, \cite{Vidal23} does not deal with ruling out undesired behaviours.
\modifRev{Reversibility, in terms of transactional behaviours, has been investigated in
the context of the actor model also in \cite{FieldV05,SwalensKM17}. While
\cite{FieldV05} introduces a facility to create a global checkpoint among multiple actors, \cite{SwalensKM17} studies the application of software transactional memory to the actor model. In \cite{FieldV05} commits and rollbacks are explicit and confined inside the scope of a stabilisation of the actor state. On the other hand, in \cite{SwalensKM17}
rollbacks and commits are implicit: if at the end the transaction can commit it does so, otherwise it is rolled back to the very beginning and re-tried again. Also, in this work reversibility is confined within the transactional scope. Let us note that \cite{FieldV05} uses a very fine-grained mechanism to keep track of the causality graph among all the actors that have interacted with the one who wants to commit. We will further study this mechanism when adapting our theory to multi-party session types.
}

\begin{table}
    \centering
  \modif{  \begin{tabular}{|c|c|c|c|c|}
    \hline
    \textbf{Work} & \textbf{Formalism}& \textbf{Commit} & \textbf{Rollback}  & \textbf{Confined}\\
    \hline
        ~\cite{LaneseMSS11} & Process Calculi &I & E & N\\
        \hline
        ~\cite{DanosK05} & Process Calculi& E & I  & Y\\
                \hline
        ~\cite{VriesKH10} & Process Calculi& E & I & N \\
                \hline 
         ~\modifRev{\cite{BarbaneraDd14, BarbaneraDLd16, BarbaneraLd17}} &Contracts& I & I & Y\\
                 \hline
        \cite{MezzinaP17a, MezzinaP21} &Session Types & I & I  & N \\
        \hline
         \cite{CastellaniDG17}& Session Types& I & I & Y \\
         \hline
\cite{MezzinaT17,fmt18,FrancalanzaMT20}&Global Graphs& I & E & Y\\
                 \hline
         \cite{Vidal23} & Process Calculi & E & E& N\\
         \hline
        \modifRev{ \cite{FieldV05}} & Actor Model & E & E & Y\\
                  \hline
        \modifRev{ \cite{SwalensKM17}} & Actor Model & I & I & Y\\
         \hline
                  \textit{Our work} & Actor Model \& & E & E& N\\
                  & Session Types & & & \\
                  \hline

    \end{tabular}
    }
    \caption{Approaches in the literature; commits and rollbacks can be either implicit (I) or explicit (E); reversibility can be confined (Y) or not (N).}
    \label{tab:overview}
\end{table}
Table \ref{tab:overview} summarises all the approaches in the literature. We detach from these works in several ways. Our checkpoint facility is explicit and checkpointing is not 
confined to choices: the programmer can decide at any point when to commit. 
This is because the programmer may be interested in committing, besides choice points, 
a series of interactions (e.g., to make a payment irreversible). 
Once a commit is taken, the system cannot revert to a state prior to it.
 Our rollback  is explicit, meaning that it is the programmer who deliberately triggers a rollback.
 %
 %
%
Our compliance check, which is decidable, resembles those \modifRev{for contracts introduced in} \cite{BarbaneraDd14,BarbaneraDLd16,BarbaneraLd17}, which are defined for 
different rollback recovery approaches based on implicit checkpoints. 
\modifRev{Specifically, our compliance relation is similar to the ones defined in 
\cite{BarbaneraDLd16,BarbaneraLd17} as they consist in requiring that, whenever no 
reduction is possible, all client requests and offers have been satisfied. Similarly, 
in our compliance notion, two participants are compliant if, whenever they reach a 
type configuration where no reduction is possible, they are in the successful (end) 
state. Our notion differs from the others as we do not distinguish client and server 
roles and, of course, we have different technicalities as our types 
(resp. type configurations) differ from retractable contracts 
(resp. client/server pairs).
The compliance relation in \cite{BarbaneraDd14}, instead, differs from the others (including ours)
as it is coinductively defined. We have not considered this approach for our definition as it 
would make the theoretical framework much more complicated and, most of all, 
would be less suitable for a \maude\ implementation of the compliance checking. 
}

Concerning the \maude\ implementation of the compliance check, we have followed 
the approach of the seminal work by Verdejo and Mart{\'\i}-Oliet \cite{CCSMaude2}, 
providing the state-of-the-art implementation of CCS in \maude. Along the same line, 
many other \maude\ implementations of formalisms and languages have been 
proposed, such as BPMN~\cite{BPMN}, Twitlang~\cite{Twitlang}, SCEL~\cite{SCEL}, and 
QFLan~\cite{QFLan}.

\section{Concluding remarks}
\label{sec:conclusion}

This paper proposes rollback recovery primitives for session-based programming. 
%
These primitives come together with session typing, enabling a design time compliance 
check  which ensures 
checkpoint persistency properties (Lemma~\ref{lemma:safe_roll} and 
Theorem~\ref{lemma:undoability}) and session soundness 
(Theorems~\ref{th:roll_soundness} and \ref{th:com_soundness}).
%
Our compliance check has been implemented in  \maude.
%

As future work, 
we plan to extend our approach to deal 
%
with sessions where parties can interleave interactions 
performed along different sessions. 
This requires to deal with subordinate sessions, which may affect enclosing sessions
by performing, e.g., commit actions that make some interaction of the enclosing sessions 
irreversible, similarly to nested transactions~\cite{Weikum92conceptsand}.
To tackle this issue it would be necessary to extend the notion of compliance 
relation  to take into account possible partial commits (in case of nested sub-sessions) that
could be undone at the top level if a rollback is performed. 
\added{Also, the way our checkpoints are taken resembles the Communication Induced Checkpoints (CIC) approach~\cite{surveycheckpoint}; 
we leave as future work a thoughtful comparison between these two mechanisms.}

\modifRev{Another possible future work is to adapt our framework to work in asynchronous settings, like real applications.
The first thing to do is to extend our framework with queues (either a global one or one queue per participant), in order to deal with asynchronous messages.
We could adapt the work of \cite{MezzinaP17a,MezzinaP21} in which uncontrolled reversibility is added to asynchronous binary and multiparty session types. 
Adding controlled reversibility atop on them would require special messages for commit, abort and checkpoint. Such messages should be handled with priority, or at least one has to assume fairness, otherwise there will be no guarantee that a commit or a rollback will be performed. We could start from the asynchronous semantics given in~\cite{LaneseMSS11}. Let us note that considering messages with priority is totally licit, as languages for large scale applications such as Erlang, Elixir and Akka allow for messages with priority.
Starting from this settings and adding a fine-grained causality tracking mechanism, such as the one  of \cite{FieldV05}, is our long-term goal, allowing us to have a complete theory for a real-world fault-tolerant applications.
}

\modifRev{Finally, we plan to investigate the practical application of our work to more realistic programming languages. Since our proposal has been devised for communication-centric systems, as a first testing ground we plan to transfer our ideas to a message-passing programming language based on the \pic, e.g.
SePi\footnote{\modifRev{\url{https://rss.rd.ciencias.ulisboa.pt/tools/sepi/}}}~\cite{sepi}.
The challenge would be not only the extension of the language with our linguistic primitives to program reversible sessions, but also extending our results to a setting richer in terms of programming constructs and features. 
Then, we will consider programming languages that feature session-based programming, but are based on paradigms other than the \pic,
e.g.
\textit{sessionj}\footnote{\modifRev{\url{https://code.google.com/archive/p/sessionj/}}}~\cite{sessionj},
an extension of Java with session-based constructs.
Another direction would be the application of our approach to Scribble\footnote{\modifRev{\url{https://github.com/scribble}}}~\cite{scribble}, a framework to specify application-level protocols among communicating systems that supports
bindings for several high-level languages.
}

\section*{Acknowledgment}
The authors would like to acknowledge anonymous reviewers  for their useful criticisms and suggestions that
have greatly contributed to improve the paper.

\bibliographystyle{alphaurl}
\bibliography{biblio}

\newpage
\appendix

\section{Proofs}
\label{proofs:binary}

\newcommand{\mFunc}[2]{tl_#1(#2)}


\modiff{The auxiliary lemmata required to prove the soundness results rely on the following definitions:}
\begin{itemize}
\item a process $\checkpointType{P}$ and  a type $\checkpointType{\typeT}$
are in \emph{checkpoint accordance} if $\checkpointType{P}=P$ implies $\checkpointType{\typeT}=\typeT$,
and $\checkpointType{P}=\imposed{P}$ implies $\checkpointType{\typeT}=\imposed{\typeT}$;
\item let $\actionLab$ a process label, its \emph{dual label} $\bar{\actionLab}$ is defined as follows:
\added{$\overline{\send{k}{v}}=\receive{k}{x}$ for some $x$,}
\added{$\overline{\receive{k}{x}}=\send{k}{v}$ for some $v$,}
$\overline{\select{k}{l}}=\branching{k}l$,
$\overline{\branching{k}l}=\select{k}{l}$; this notion of duality straightforwardly extends to type labels;
\item the function $\mFunc{\sorting}{\cdot}$, mapping process labels to type labels under sorting $\sorting$, is defined as follows:
$\mFunc{\sorting}{\send{k}{v}}=\outType{S}$ with $\sorting \judge v \hasType S$,
$\mFunc{\sorting}{\receive{k}{x}}=\inpType{S}$ with $\sorting \judge x \hasType S$,
$\mFunc{\sorting}{\select{k}{l}}=\selTypeLabel{l}$,
$\mFunc{\sorting}{\branching{k}l}=\branchTypeLabel{l}$,
$\mFunc{\sorting}{\commitLab}=\commitType$,
$\mFunc{\sorting}{\rollLab}=\rollType$,
$\mFunc{\sorting}{\abortLab}=\abortType$, and
$\mFunc{\sorting}{\ite}=\tau$.
\end{itemize}

The following lemma states that each reduction of a reachable collaboration corresponds 
to a reduction of its configuration types.
\begin{lem}\label{lemma:subj}
Let $C=\singleSession{s}{C'}(\logged{{\genSession}}{\checkpointType{P_1}} P_2
\ \mid \ 
\logged{\genSession}{\checkpointType{Q_1}} Q_2)$ 
be a reachable collaboration,
with 
$C'=(\requestAct{a}{x}{P} \mid \acceptAct{a}{y}{Q})$,
\mbox{$(\emptyset;\emptyset \judge P \hasType x:\typeT)$},
\mbox{$(\emptyset;\emptyset \judge Q \hasType y:\typeT')$}, 
\mbox{$\typeT \compliant \typeT'$},
\mbox{$(\basis_1;\sorting_1 \judge P_1\subst{x}{\ce{s}} \hasType x:\typeT_1)$},
\mbox{$(\basis_2;\sorting_2 \judge P_2\subst{x}{\ce{s}} \hasType x:\typeT_2)$},
{$(\basis_1';\sorting_1' \judge Q_1\subst{y}{{s}} \hasType y:\typeU_1)$}, and 
\mbox{$(\basis_2';\sorting_2' \judge Q_2\subst{y}{{s}} \hasType y:\typeU_2)$}.
If $C \fwbwred \singleSession{s}{C'}(\logged{{\genSession}}{\checkpointType{P_1'}} P_2'
\ \mid \ 
\logged{\genSession}{\checkpointType{Q_1'}} Q_2')$
then there exist $\typeT_1'$, $\typeT_2'$, $\typeU_1'$, $\typeU_2'$ such that
$\initConf{\typeT}{\typeT'}\conf{\checkpointType{\typeT_1}}{\typeT_2} \confcomp \conf{\checkpointType{\typeU_1}}{\typeU_2}
\typered
\initConf{\typeT}{\typeT'}\conf{\checkpointType{\typeT_1'}}{\typeT_2'} 
\confcomp \conf{\checkpointType{\typeU_1'}}{\typeU_2'}$ 
with 
$\checkpointType{P_1'}$ (resp. $\checkpointType{Q_1'}$) in checkpoint accordance 
with $\checkpointType{\typeT_1'}$ (resp. $\checkpointType{\typeU_1'}$), 
{$(\hat{\basis}_1;\hat{\sorting}_1 \judge P_1'\subst{x}{\ce{s}} \hasType x:\typeT_1')$},
{$(\hat{\basis}_2;\hat{\sorting}_2 \judge P_2'\subst{x}{\ce{s}} \hasType x:\typeT_2')$},
{$(\hat{\basis}_1';\hat{\sorting}_1' \judge Q_1'\subst{y}{{s}} \hasType y:\typeU_1')$}, and 
{$(\hat{\basis}_2';\hat{\sorting}_2' \judge Q_2'\subst{y}{{s}} \hasType y:\typeU_2')$}.
\end{lem}
\begin{proof}
We have two cases depending whether the reduction $\fwbwred$
has forward or backward direction. 
\begin{description}
\item[$(\fwbwred \,=\, \fwred)$]
From rule \rulelabel{F-Res}, we have 
$\logged{{\genSession}}{\checkpointType{P_1}} P_2
\ \mid \ \logged{\genSession}{\checkpointType{Q_1}} Q_2
\fwred 
\logged{{\genSession}}{\checkpointType{P_1'}} P_2'
\ \mid \ \logged{\genSession}{\checkpointType{Q_1'}} Q_2'$.
We prove the result by case analysis on the last rule applied in 
the inference of the above reduction.
\begin{itemize}
\item \rulelabel{F-Com}. 
In this case we have $P_2 \auxrel{\send{\ce{s}}{v}} P_2'$,
$Q_2 \auxrel{\receive{s}{z}} Q_2''$, with $Q_2'=Q_2''\subst{v}{z}$,
$\checkpointType{P_1'}=\checkpointType{P_1}$ and 
$\checkpointType{Q_1'}=\checkpointType{Q_1}$.
Thus, $P_2\subst{x}{\ce{s}}=\sendAct{x}{e}{P_2'\subst{x}{\ce{s}}}$
for some $e$ such that $\expreval{e}{v}$, and 
$Q_2\subst{y}{{s}}=\receiveAct{y}{z:S'}{Q_2'\subst{y}{{s}}}$.
By rule \rulelabel{T-Snd}, we have that $\typeT_2=\outType{S}.\typeT_2'$,
with $\sorting_2 \judge e \hasType S$, (hence $\sorting_2 \judge v \hasType S$),
and $\basis_2;\sorting_2 \judge P_2'\subst{x}{\ce{s}} \hasType x:\typeT_2'$
(hence $\hat{\basis}_2 = \basis_2$ and $\hat{\sorting}_2 = \sorting_2$).
Similarly, by rule \rulelabel{T-Rcv}, we have that 
$\typeU_2=\inpType{S'}.\typeU_2'$ and 
$\basis_2';\sorting_2'\comp z:S' \judge Q_2''\subst{y}{{s}} \hasType y:\typeU_2'$
(hence $\hat{\basis}_2' = \basis_2'$ and $\hat{\sorting}_2' = \sorting_2'\comp z:S'$).
By rules \rulelabel{TS-Snd} and \rulelabel{TS-Rcv}, we get
$\typeT_2 \typeTrans{\outType{S}} \typeT_2'$ and 	
$\typeU_2 \typeTrans{\inpType{S'}} \typeU_2$.
Now, reasoning by contradiction, let us suppose that $S\neq S'$. Thus, the term 
$\initConf{\typeT}{\typeT'}\conf{\checkpointType{\typeT_1}}{\typeT_2} \confcomp \conf{\checkpointType{\typeU_1}}{\typeU_2}
\typered\!\!\!\!\!\!\!\!\!/\ \ \ $, since no rule in Fig.~\ref{fig:typeSemantics_ext}
can be applied. However, since $C$ is a reachable collaboration, 
this type configuration is originated from 
$\initConf{\typeT}{\typeT'}\conf{{\typeT}}{\typeT} \confcomp \conf{{\typeT'}}{\typeT'}$.
By Def.~\ref{def:compliance}, $\typeT \compliant \typeT'$ implies 
$\typeT_2=\typeU_2=\inactType$, which is a contradiction.
Therefore, it holds that $S=S'$. Hence, by applying rule \rulelabel{TS-Com} we can 
conclude.
\item \rulelabel{F-Lab}, \rulelabel{E-Cmt$_1$} and \rulelabel{E-Cmt$_2$}.
Similar to the previous case.
\item \rulelabel{F-Par}.
In this case we have that
$\logged{{\genSession}}{\checkpointType{P_1}} P_2
\fwred 
\logged{{\genSession}}{\checkpointType{P_1'}} P_2'$.
Since this transition involves only one log term, it can be inferred only by applying rule \rulelabel{F-If},
from which we have $P_2 \auxrel{\ite} P_2'$ and $\checkpointType{P_1'}=\checkpointType{P_1}$.
By rule \rulelabel{P-IfT} (the case of rule \rulelabel{P-IfF} is similar), we have $P_2\subst{x}{\ce{s}}=
\ifthenelseAct{e}{P_2'\subst{x}{\ce{s}}}{R}$ with $\expreval{e}{\ctrue}$.
By rule \rulelabel{T-If} we get $\typeT_2=\typeT_2' \choiceType \typeV$
and $\basis_2;\sorting_2 \judge P_2'\subst{x}{\ce{s}} \hasType x:\typeT_2'$.
By rule \rulelabel{TS-If$_1$}, $\typeT_2 \typeTrans{\tau} \typeT_2'$.
By applying rule \rulelabel{TS-Tau} we can conclude.
\end{itemize}

\item[$(\fwbwred \,=\, \bwred)$] 
\modiff{The reduction can be inferred by applying rule \rulelabel{B-Res} or 
rule \rulelabel{B-Abt}. Let us consider the former case, the latter is similar.}
From rule \rulelabel{B-Res}, we have 
$\logged{{\genSession}}{\checkpointType{P_1}} P_2
\ \mid \ \logged{\genSession}{\checkpointType{Q_1}} Q_2
\bwred 
\logged{{\genSession}}{\checkpointType{P_1'}} P_2'
\ \mid \ \logged{\genSession}{\checkpointType{Q_1'}} Q_2'$.
We prove the result by case analysis on the last rule applied in 
the inference of the above reduction.
\begin{itemize}
\item \rulelabel{B-Rll}. 
In this case $P_2 \auxrel{\rollLab} P_2'$,
$\checkpointType{P_1'}=\checkpointType{P_1}$,
$\checkpointType{Q_1'}= \checkpointType{Q_1}$
and $Q_2'=Q_1$.
By rule \rulelabel{P-Rll}, we have $P_2\subst{x}{\ce{s}}=\roll$
and $P_2'\subst{x}{\ce{s}}=\inact$. 
By rule \rulelabel{T-Rll} we get $\typeT_2=\rollType$.
\modiff{By hypothesis $\typeT \compliant \typeT'$, which implies that 
rule \rulelabel{TS-Rll$_2$} is not applicable, because by Def.~\ref{def:compliance} an erroneous type configuration cannot be reached. Hence, by rule \rulelabel{TS-Rll$_1$},} $\typeT_2 \typeTrans{roll} \typeT_2'$,
with $\typeT_2'=\inactType$.
By rule \rulelabel{T-Inact}, we have 
$\basis_2;\sorting_2 \judge P_2'\subst{x}{\ce{s}} \hasType x:\typeT_2'$.		
Finally, by applying rule \rulelabel{TS-Rll$_1$} we can conclude.
\item 
\rulelabel{B-Par}. Similarly to the forward case. \qedhere
\end{itemize}
\end{description}
\end{proof}

The following lemma relates collaboration reductions to type reductions when 
a \linebreak[5] $\rollError$ is produced.
\begin{lem}\label{lemma:typeCorrespondence}
Let $C=(\requestAct{a}{x}{P} \mid \acceptAct{a}{y}{Q})$ such that
$C \hasType \{\ce{a}:\typeT_1,a:\typeT_2\}$
and 
$\typeT_1 \compliant \typeT_2$.
If $C \fwbwred^* \singleSession{s}{C}(\logged{\ce{\genSession}}{\imposed{P_1}} P_2 
\ \mid \ \logged{\genSession}{\checkpointType{Q_1}} Q_2)$
and $P_2 \auxrel{\rollLab} P_2'$, 
then there exist
$\typeU_1$, $\typeU_2$, $\typeU_1'$, $\typeU_2'$, $\typeU_1''$ such that
$\initConf{\typeT_1}{\typeT_2}\conf{\typeT_1}{\typeT_1} \confcomp \conf{\typeT_2}{\typeT_2}
\typered^*
\initConf{\typeT_1}{\typeT_2}\conf{\imposed{\typeU_1}}{\typeU_1'} 
\confcomp \conf{\checkpointType{\typeU_2}}{\typeU_2'}$
and 
$\typeU_1' \typeTrans{roll} \typeU_1''$.
\end{lem}
\begin{proof}
From $C \hasType \{\ce{a}:\typeT_1,a:\typeT_2\}$, by applying 
\rulelabel{T-Par}, \rulelabel{T-Acc} and \rulelabel{T-Req}, 
we have that  
$\emptyset;\emptyset \judge P \hasType x:\typeT_1$
and 
$\emptyset;\emptyset \judge Q \hasType x:\typeT_2$.
By applying rule \rulelabel{F-Con} to the collaboration $C$, we obtain
$C \fwred \singleSession{s}{C}
(\logged{\ce{s}}{P\subst{\ce{s}}{x} } P\subst{\ce{s}}{x}
        \ \mid\ 
	\logged{s}{Q\subst{s}{y}} Q\subst{s}{y} ) = C'
$.
Now, by repeatedly applying Lemma~\ref{lemma:subj}, from 
$C' \fwbwred^* \singleSession{s}{C}(\logged{\ce{\genSession}}{\imposed{P_1}} P_2 
\ \mid \ \logged{\genSession}{\checkpointType{Q_1}} Q_2)$ 
we get 
$\initConf{\typeT_1}{\typeT_2}\conf{\typeT_1}{\typeT_1} \confcomp \conf{\typeT_2}{\typeT_2}
\typered^*
\initConf{\typeT_1}{\typeT_2}\conf{\imposed{\typeU_1}}{\typeU_1'} 
\confcomp \conf{\checkpointType{\typeU_2}}{\typeU_2'}$
for some $\typeU_1$, $\typeU_2$, $\typeU_1'$, $\typeU_2'$, 
with 
$\basis_2;\sorting_2 \judge P_2\subst{x}{\ce{s}} \hasType x:\typeU_1'$.
Now, let us consider the transition $P_2 \auxrel{\rollLab} P_2'$.
This can be derived only by the application of rules \rulelabel{P-Rll}.
Thus, $P_2 = \roll$, from which we have $P_2\subst{x}{\ce{s}}=\roll$.
From $\basis_2;\sorting_2 \judge \roll \hasType x:\typeU_1'$, by rule 
\rulelabel{T-Rll}, we get $\typeU_1'=\rollType$. Therefore, 
by \modiff{rule \rulelabel{TS-Rll$_2$}}, 
we can conclude $\typeU_1' \typeTrans{roll} \typeU_1''$
with $\typeU_1''=\inactType$.
\end{proof}

The following lemma relates collaboration reductions to type reductions when 
a \linebreak[5] $\comError$ is produced.
\begin{lem}\label{lemma:typeCorrespondenceCom}
Let $C=(\requestAct{a}{x}{P} \mid \acceptAct{a}{y}{Q})$ such that
$C \hasType \{\ce{a}:\typeT_1,a:\typeT_2\}$
and 
$\typeT_1 \compliant \typeT_2$.
If $C \fwbwred^* \singleSession{s}{C}(\logged{\ce{\genSession}}{\checkpointType{P_1}} P_2 
\ \mid \ \logged{\genSession}{\checkpointType{Q_1}} Q_2)$,
$P_2 \auxrel{\actionLab} P_2'$ and
$\neg Q_2 \Barb{\bar{\actionLab}}$
with $\actionLab$ of the form 
$\send{k}{v}$,
$\receive{k}{x}$,
$\select{k}{l}$ or
$\branching{k}l$, 
then there exist
$\typeU_1$, $\typeU_2$, $\typeU_1'$, $\typeU_2'$, $\typeU_1''$ such that
$\initConf{\typeT_1}{\typeT_2}\conf{\typeT_1}{\typeT_1} \confcomp \conf{\typeT_2}{\typeT_2}
\typered^*
\initConf{\typeT_1}{\typeT_2}\conf{\checkpointType{\typeU_1}}{\typeU_1'} 
\confcomp \conf{\checkpointType{\typeU_2}}{\typeU_2'}$ with
$\checkpointType{P_1}$ (resp. $\checkpointType{Q_1}$)
in checkpoint accordance with $\checkpointType{\typeU_1}$ 
(resp. $\checkpointType{\typeU_2}$),
$\typeU_1' \typeTrans{\mFunc{\sorting}{\actionLab}} \typeU_1''$, with
$\sorting$ sorting for typing $P_2'$,
and for all $\typeU_2''$ such that $\typeU_2' \typeTrans{\tau}^* \typeU_2''$
we have 
$\typeU_2'' \typeTrans{\overline{\mFunc{\sorting}{{\actionLab}}}}\!\!\!\!\!\!/\ $.
\end{lem}
\begin{proof}
From $C \hasType \{\ce{a}:\typeT_1,a:\typeT_2\}$, by applying 
\rulelabel{T-Par}, \rulelabel{T-Acc} and \rulelabel{T-Req}, 
we have that  
$\emptyset;\emptyset \judge P \hasType x:\typeT_1$
and 
$\emptyset;\emptyset \judge Q \hasType x:\typeT_2$.
By applying rule \rulelabel{F-Con} to the collaboration $C$, we obtain
$C \fwred \singleSession{s}{C}
(\logged{\ce{s}}{P\subst{\ce{s}}{x} } P\subst{\ce{s}}{x}
        \ \mid\ 
	\logged{s}{Q\subst{s}{y}} Q\subst{s}{y} ) = C'
$.
Now, by repeatedly applying Lemma~\ref{lemma:subj}, from 
$C' \fwbwred^* \singleSession{s}{C}(\logged{\ce{\genSession}}{\checkpointType{P_1}} P_2 
\ \mid \ \logged{\genSession}{\checkpointType{Q_1}} Q_2)$
we get 
$\initConf{\typeT_1}{\typeT_2}\conf{\typeT_1}{\typeT_1} \confcomp \conf{\typeT_2}{\typeT_2}
\typered^*
\initConf{\typeT_1}{\typeT_2}\conf{\checkpointType{\typeU_1}}{\typeU_1'} 
\confcomp \conf{\checkpointType{\typeU_2}}{\typeU_2'}$
for some $\typeU_1$, $\typeU_2$, $\typeU_1'$, $\typeU_2'$, 
with 
$\basis_2;\sorting_2 \judge P_2\subst{x}{\ce{s}} \hasType x:\typeU_1'$
and 
$\basis_2';\sorting_2' \judge Q_2\subst{y}{{s}} \hasType y:\typeU_2'$.
Now, let us reason by case analysis on the rule for deriving the transition $P_2 \auxrel{\actionLab} P_2'$.
\begin{description}
\item[Rule \rulelabel{P-Snd}]
Thus, $P_2 = \sendAct{\ce{s}}{e}{P_2'}$ and $\actionLab=\send{\ce{s}}{v}$
with $\expreval{e}{v}$. 
From $\basis_2;\sorting_2 \judge  P_2\subst{x}{\ce{s}}  \hasType x:\typeU_1'$, by rule 
\rulelabel{T-Snd}, we get $\typeU_1'=\outType{S}.\typeU_1''$ with 
$\sorting \judge e \hasType S$. Therefore, by rule \rulelabel{TS-Snd}, we get 
$\typeU_1' \typeTrans{\outType{S}} \typeU_1''$.

\item[Rule \rulelabel{P-Rcv}]
Thus, $P_2 = \receiveAct{\ce{s}}{y:S}{P_2'}$ and $\actionLab=\receive{\ce{s}}{y}$. 
From $\basis_2;\sorting_2 \judge  P_2\subst{x}{\ce{s}}  \hasType x:\typeU_1'$, by rule 
\rulelabel{T-Rcv}, we get $\typeU_1'=\inpType{S}.\typeU_1''$. Therefore, by rule \rulelabel{TS-Rcv}, we get 
$\typeU_1' \typeTrans{\inpType{S}} \typeU_1''$.

\item[Rules \rulelabel{P-Sel} and \rulelabel{P-Brn}] Similar to the previous cases.

\end{description}
Finally, from $\neg Q_2 \Barb{\bar{\actionLab}}$, following a similarly reasoning, we can 
conclude $\typeU_2' \typeTrans{\tau}^* \typeTrans{\overline{\mFunc{\sorting}{{\actionLab}}}}\!\!\!\!\!\!/\ $.
\end{proof}

\bigskip

\modiff{The soundness results are as follows.}

\bigskip

\theoremApp{\ref{th:roll_soundness}}
If $C$ is a roll-safe collaboration, then we have that $C\ {\fwbwred\!\!\!\!\!\!/\ }^*\ \collContext[\rollError]$.
\begin{proof} 
The proof proceeds by contradiction.
Suppose that there exists an initial collaboration $C$ that is rollback safe and such that 
$C \fwbwred^* \collContext[\rollError]$. 
The erroneous collaboration $\rollError$ can be only produced by applying rule \rulelabel{E-Rll$_2$}.
Thus, to infer at least one reduction of the sequence $C \fwbwred^* \collContext[\rollError]$,
rule \rulelabel{E-Rll$_2$} must be used. From this, we have that there exists a runtime 
collaboration $C' \congr \collContext'[C'']$, with 
$C''=(\logged{\ce{\genSession}}{\imposed{Q_1}} P_1 
\ \mid \ \logged{\genSession}{\checkpointType{Q_2}} P_2)$, such that 
$C \fwbwred^* C'$,
$P_1 \auxrel{\rollLab} P_1'$, 
and 
$C' 	\ \fwred\  \collContext'[\rollError] \ \fwbwred^*\ \collContext[\rollError]$.
By rules \rulelabel{F-Con}, \rulelabel{F-Res} and \rulelabel{B-Res},   
and the fact that the scope of operator $\singleSession{s}{\_}$ is statically defined
(i.e., neither the operational rules nor $\congr$ allow scope extension), 
the term $C''$ can only be the argument of the operator $\singleSession{s}{C_1}$,
i.e. $\collContext'=C_2\ \mid \ \singleSession{s}{C_1}[\cdot]$, 
with $C_1=\requestAct{a}{x}{P} \mid \acceptAct{a}{y}{Q}$ for some 
$a$, $x$, $y$, $P$ and $Q$.
In its own turn, the term $\singleSession{s}{C_1}C''$ can only be generated by 
applying rule \rulelabel{F-Con} from $C_1$, which must be a subterm of $C$,
i.e. $C \congr C_1 \mid C_2'$ for some $C_2'$.
Since the scope of $\singleSession{s}{\_}$ operator cannot be extended,
all reductions performed by terms in parallel with it by applying rules \rulelabel{F-Par}
and \rulelabel{B-Par} do not affect the argument of such operator.
Therefore, focussing on the subterm $C_1$ of $C$, by exploiting rules  
\rulelabel{F-Par} and \rulelabel{B-Par} we can set apart the reductions in 
$C \fwbwred^* \collContext[\rollError]$ involving $C_1$ and its derivatives,
thus obtaining $C_1 \fwbwred^* \singleSession{s}{C_1}C''  \fwred \singleSession{s}{C_1}\rollError$.

Now, since $C$ is rollback safe, by Def.~\ref{def:rollback_safety} we have that $C \hasType \sessions$
and for all pairs $\ce{b}:\typeV_1$ and $b:\typeV_2$ in $\sessions$ we have 
$\typeV_1 \compliant \typeV_2$. 
Since $C \congr C_1 \mid C_2'$, by rule \rulelabel{T-Par} we obtain 
$C_1 \hasType \sessions_1$ with $\sessions_1 \subseteq \sessions$.
By rules \rulelabel{T-Req} and \rulelabel{T-Acc}, we have  
$\sessions_1 = \{\ce{a}:\typeT_1,a:\typeT_2\}$. 
Since \modiff{$\sessions_1$ is a subset of $\sessions$,} we have that $\typeT_1 \compliant \typeT_2$.

By Lemma~\ref{lemma:typeCorrespondence}, we have that there exist
$\typeU_1$, $\typeU_2$, $\typeU_1'$, $\typeU_2'$, $\typeU_1''$ such that
$\initConf{\typeT_1}{\typeT_2}\conf{\typeT_1}{\typeT_1} \confcomp \conf{\typeT_2}{\typeT_2}
\typered^*
\initConf{\typeT_1}{\typeT_2}\conf{\imposed{\typeU_1}}{\typeU_1'} 
\confcomp \conf{\checkpointType{\typeU_2}}{\typeU_2'}
=t
$
and 
$\typeU_1' \typeTrans{roll} \typeU_1''$.
%
%
\added{Since $\typeU_1'$ can only perform $roll$, \modiff{the only rules that can be applied are \rulelabel{TS-Rll$_1$} and \rulelabel{TS-Rll$_2$}. However,} rule \rulelabel{TS-Rll$_1$} cannot be applied due to the imposed checkpoint $\imposed{\typeU_1}$. 
Therefore, the only rule that can be applied is 
\rulelabel{TS-Rll$_2$}, leading to the configuration 
$
t'= 
\initConf{\typeT_1}{\typeT_2}\conf{\imposed{\typeU_1}}{\typeU_1''} 
\confcomp \conf{\checkpointType{\typeU_2}}{\typeU_2''}
$
with $\typeU_1''=\typeU_2''=\errType$.
Now, no rule in Fig.~\ref{fig:typeSemantics_ext} allows the term $t'$ to evolve, i.e. $t' \typered\!\!\!\!\!\!\!\!\!/\ \ \ $.
Since $\typeT_1 \compliant \typeT_2$, by Def.~\ref{def:compliance} it must hold that $\typeU_1'' = \typeU_2'' = \inactType$. However,  $\typeU_1'' = \errType \neq \inactType$ and $\typeU_2'' = \errType  \neq \inactType$, which is a contradiction.}
\end{proof}

\bigskip
\theoremApp{\ref{th:com_soundness}}
If $C$ is a roll-safe collaboration, then we have that $C\ {\fwbwred\!\!\!\!\!\!/\ }^*\ \collContext[\comError]$.
\begin{proof} 
The proof proceeds by contradiction.
Suppose that there exists an initial collaboration $C$ that is rollback safe and such that 
$C \fwbwred^* \collContext[\comError]$. 
The erroneous collaboration $\comError$ can be produced by applying one of 
the rules \rulelabel{E-Com1}, \rulelabel{E-Com2}, \rulelabel{E-Lab1} and \rulelabel{E-Lab2}.
Let us consider the case \rulelabel{E-Com1}, the other cases are similar.
Proceeding as in the proof of Theorem~\ref{th:roll_soundness}, without loss of generality 
we can focus on the subterm $C_1=(\requestAct{a}{x}{P} \mid \acceptAct{a}{y}{Q})$ 
of $C$, such that 
$C_1 \fwbwred^* \singleSession{s}{C_1}C''  \fwred \singleSession{s}{C_1}\comError$,
with $C''=(\logged{{\genSession}}{\checkpointType{Q_1}} P_1 
\mid \logged{\genSession}{\checkpointType{Q_2}} P_2)$,
and $C_1 \hasType \{\ce{a}:\typeT_1,a:\typeT_2\}$, with $\typeT_1 \compliant \typeT_2$. 

By Lemma~\ref{lemma:typeCorrespondenceCom}, we have that there exist
$\typeU_1$, $\typeU_2$, $\typeU_1'$, $\typeU_2'$, $\typeU_1''$ such that
$\initConf{\typeT_1}{\typeT_2}\conf{\typeT_1}{\typeT_1} \confcomp \conf{\typeT_2}{\typeT_2}
\typered^*
\initConf{\typeT_1}{\typeT_2}\conf{\checkpointType{\typeU_1}}{\typeU_1'} 
\confcomp \conf{\checkpointType{\typeU_2}}{\typeU_2'}=t$ with
$\checkpointType{Q_1}$ (resp. $\checkpointType{Q_2}$)
in checkpoint accordance with $\checkpointType{\typeU_1}$ 
(resp. $\checkpointType{\typeU_2}$),
$\typeU_1' \typeTrans{\outType{S}} \typeU_1''$, 
and for all $\typeU_2''$ such that $\typeU_2' \typeTrans{\tau}^* \typeU_2''$
we have 
$\typeU_2'' \typeTrans{\inpType{S}}\!\!\!\!\!\!/\ $.
Thus, for all $\typeU_2''$ as above, we have
$t \typered^* 
\initConf{\typeT_1}{\typeT_2}\conf{\checkpointType{\typeU_1}}{\typeU_1'} 
\confcomp \conf{\checkpointType{\typeU_2}}{\typeU_2''}=t'$.
No rule in Fig.~\ref{fig:typeSemantics_ext} allows the term $t'$ to evolve, i.e. $t' \typered\!\!\!\!\!\!\!\!\!/\ \ \ $,
because $\typeU_1'$ can only perform $\outType{S}$ and rule \rulelabel{TS-Com} cannot be applied 
since $\typeU_2'' \typeTrans{\inpType{S}}\!\!\!\!\!\!/\ $.
Since $\typeT_1 \compliant \typeT_2$, by Def.~\ref{def:compliance} it must hold that 
$\typeU_1' = \inactType$.
However, since $\typeU_1'$ is able to perform an action (as it holds that 
$\typeU_1' \typeTrans{\outType{S}} \typeU_1''$), we get that 
it cannot be an end type, i.e. $\typeU_1' \neq \inactType$, which is a contradiction. 
\end{proof}

\bigskip

\newcommand{\corollaryApp}[2]{
\noindent
\textbf{Corollary #1.} \emph{#2}}




\end{document}